\documentclass[10pt]{amsart}
\usepackage{multirow}
\usepackage{booktabs}
\usepackage{amsfonts}
\usepackage{amssymb}
\usepackage{amsthm}
\usepackage{color}
\usepackage{latexsym}
\usepackage{bbm}
\usepackage{enumerate}
\usepackage{graphicx}
\usepackage{graphics}
\usepackage{float}
\usepackage{sidecap}
\usepackage{wrapfig}
\usepackage{subfig}
\usepackage{graphicx}
\usepackage{amsmath}
\usepackage{setspace}
\usepackage[tmargin=2.0cm, bmargin=2.0cm, lmargin=2.0cm, rmargin=2.0cm]{geometry}
\numberwithin{equation}{section}

\def\Wp{W^{(q)\prime}}

\def\P{{\mathbb P}}
\def\E{{\mathbb E}}

  \def\i{\infty}

\newcommand*\diff{\mathop{}\!\mathrm{d}}
\def\nn{\nonumber}

\newtheorem{Thm}{Theorem}
\newtheorem{Lem}{Lemma}
\newtheorem{Cor}{Corollary}
\newtheorem{Prop}{Proposition}

\theoremstyle{definition}
\newtheorem{Rem}{Remark}

\DeclareFontFamily{U}{mathx}{\hyphenchar\font45}
\DeclareFontShape{U}{mathx}{m}{n}{
      <5> <6> <7> <8> <9> <10>
      <10.95> <12> <14.4> <17.28> <20.74> <24.88>
      mathx10
      }{}
\DeclareSymbolFont{mathx}{U}{mathx}{m}{n}
\DeclareFontSubstitution{U}{mathx}{m}{n}
\DeclareMathAccent{\widecheck}{0}{mathx}{"71}
\DeclareMathAccent{\wideparen}{0}{mathx}{"75}

\title{Double continuation regions for American and Swing options with negative discount rate in L\'evy models}

\author{Marzia De Donno}
\address{Department of Economics and Management, University of Parma,  via Kennedy 6, 43125  Parma, Italy}
\email{marzia.dedonno@unipr.it}

\author{Zbigniew Palmowski}
\address{Faculty of Pure and Applied Mathematics, Wroc\l aw University of Science and Technology, ul. Wyb. Wyspia\'nskiego 27, 50-370 Wroc\l aw, Poland}
\email{zbigniew.palmowski@gmail.com}

\author{Joanna Tumilewicz}
\address{Mathematical Institute, University of Wroc\l aw, pl. Grunwaldzki 2/4, 50-384 Wroc\l aw, Poland}
\email{joanna.tumilewicz@gmail.com}

\thanks{Supported by the National Science Centre under the grants 2016/23/B/HS4/00566 (2017-2020) and 2016/23/N/ST1/01189 (2017-2019).}
\date{\today}


\begin{document}

\begin{abstract}
In this paper we study perpetual American call and put options  in an exponential L\'evy model. We consider a negative effective discount rate which arises in a number of financial applications
including stock loans and real options, where the strike price can potentially grow at a higher rate than
the original discount factor. We show that in this case  a double continuation region arises and we  identify the two critical prices. We also generalize this result to multiple stopping problems of Swing type, that is, when
successive exercise opportunities are separated by i.i.d. random
refraction times. We conduct an extensive numerical analysis for the Black-Scholes model and
the jump-diffusion model with exponentially distributed jumps.

\vspace{3mm}


\end{abstract}

\maketitle

\pagestyle{myheadings} \markboth{\sc M.\ De\ Donno -- Z.\ Palmowski
--- J.\ Tumilewicz} {\sc The double continuation region for the American options}

\vspace{1.8cm}

\tableofcontents

\newpage

\section{Introduction}\label{sec:intro}
In this paper, we study the optimal stopping problem
\begin{equation}\label{mainprice} \sup_{\tau \in \mathcal{T}} \E\left[e^{-q\tau}G(\widetilde{S}_\tau)\right]
\end{equation}
for the asset price
\begin{equation*}\label{mainprice2}
\widetilde{S}_t=e^{\widetilde{X}_t},
\end{equation*}
where $\widetilde{X}$  is an asymmetric L\'evy process, $G(x) = (K-x)^+$ or $G(x) = (x-K)^+$, $\mathcal{T}$ is a family of stopping times with respect to the right-continuous augmentation of the
natural filtration of $\tilde{X}$ satisfying usual conditions,
but, in contrast to the standard assumption, we take $q<0$.
When $\tau=+\infty$, corresponding to the case where early stopping is suboptimal,
we use the convention that $e^{-q\tau}\cdot 0=0$.
In a financial market context, the solution of this problem is the price of a perpetual American option, in {\it a L\'evy market} with {\it a negative interest rate}.  Our analysis is then extended to a more general type of American options, called \emph{Swing options}, which allow for  multiple exercise opportunities separated by i.i.d. random refraction times.

The non-standard assumption of negative discount rate has received some attention in the last few years because several decision-making problems in finance fall under that umbrella.

One of the main example concerns {\it gold loans}.
After the financial crisis, 
collateralized borrowing has increased. Treasury bonds and stocks
are the collateral usually accepted by financial institutions, but gold is increasingly being used around the
world; see \cite{7}. Major Indian non-banking financial companies, like Muthoot Finance and Manappuram Finance, have
been quite active in lending against gold collateral. As 
\cite{CS} report in their survey of
the Indian gold loan market, gold loans tend to have short maturities and rather high spreads (borrowing rate
minus risk-free rate), even if significantly lower than without collateral.
The prepayment option is common,
permitting the redemption of the gold at any time before maturity.
In a gold loan, a borrower receives at time $0$ (the date of contract inception) a loan amount $K>0$ using
one mass unit (one troy ounce, say) of gold as collateral, which must be physically delivered to a lender. This
amount grows at the constant borrowing rate
$\gamma$,  stipulated
in the contract (and usually higher than the risk-free rate $r$). The cost of reimbursing the loan at time $t$ is thus given by $Ke^{\gamma t}$. When paying back the
loan, the borrower redeems the gold and the contract is terminated. In our model, we assume that the costs of storing and
insuring gold holdings are $\bar{S}_tc>0$ per unit of time, where $\bar{S}_t$ is the gold spot price, which dynamics is governed by an exponential L\'evy process; specifically,
$$\bar{S}_t=e^{\bar{X}_t}$$ where  $\bar{X}_t$ is  a L\'evy process, that is a c\`{a}dl\`{a}g process
whose increments in non--overlapping
time intervals are independent and stationary.
The dynamics of $\bar{S}_t$ under the risk-neutral measure is such that
the discounted price $e^{-rt}\bar{S}_t$ is a martingale, that is
$\E \bar{S}_t = e^{rt}\E \bar{S}_0$; see e.g. \cite{Hull}.
If, as a first step in the analysis, we  assume  that redemption can be done at any time (perpetual contract),  the value of the contract, with infinite maturity date,  at time $0$ is
\begin{equation*}\label{pricegold}V_c(s):=\sup_{\tau \in \mathcal{T}} \E\left[e^{-r\tau}\left(\bar{S}_te^{ct}-Ke^{\gamma \tau}\right)^+|\ \bar{S}_0=s\right]=
\sup_{\tau \in \mathcal{T}} \E\left[\left.e^{-q\tau}\left(\widetilde{S}_t-K\right)^+\right|\ \widetilde{S}_0=s\right]\end{equation*}
for $q=r-\gamma$, $$\widetilde{S}_t=\bar{S}_t e^{-\gamma t +ct}=e^{\widetilde{X}_t}$$ and $\widetilde{X}_t=\bar{X}_t-\gamma t+ct$ being a L\'evy process satisfying
\begin{equation}\label{mainprice3}
\E e^{\widetilde{X}_t}=e^{(r-\gamma+c)t}\E e^{\widetilde{X}_0}=e^{(q+c)t}\E e^{\widetilde{X}_0}.
\end{equation}
Thus this price equals the initial value of a perpetual American call option, that is
\eqref{mainprice} with $G(x)=(x-K)^+$, calculated with respect to the martingale measure, hence satisfying \eqref{mainprice3}.
Data from the 
\cite{10} show that the daily log change in the gold spot price,
expressed in Indian rupees, has registered an annualized historical volatility of $21.4$\% over the period from the
3rd of January 1979 to the 5th of May 2013. Average storage/insurance costs are about $2$\% and other parameters can be chosen as follows
$r = 8$\%, $c = 2$\% and $\gamma= 17$\%.
Then, in this case,
we have
\begin{equation*}\label{mainprice4}
q=r-\gamma<0,
\end{equation*} namely, a negative discounting rate.

There are many other financial products where a {\it negative rate} appears.
For example, 
\cite{Xia&Zhou} consider {\it a stock loan}, which is similar to the above gold loan, but here
a stock is used as collateral and can be redeemed at any time if it is convenient for the borrower. Such a contract can be reduced to a perpetual American put option with a possibly negative interest rate, given by the difference between the market interest rate and the loan rate.
Indeed, if an investor lends at time $0$ an amount $K$ with interest rate $\gamma$ and he/she gets as collateral some stock $\bar{S}$ with the price $\bar{S}_0$, then
the value of stock loan is given by
\begin{align}
V_p(s):=V(s):=\sup\limits_{\tau\in\mathcal{T}}\E\left[e^{-r\tau}\left(Ke^{\gamma\tau}-\bar{S}_{\tau}\right)^+|\ \bar{S}_0=s\right],\label{stockloan}
\end{align}
where $r$ is a risk-free interest rate.
The essential difference between the above instrument and a classical American put option is the time-dependent strike price.
We can still rewrite problem \eqref{stockloan} into an American put option form as follows:
\begin{align}
V(s)=&\sup\limits_{\tau\in\mathcal{T}}\E\left[e^{-r\tau}\left(Ke^{\gamma\tau}-\bar{S}_{\tau}\right)^+|\ \bar{S}_0=s\right]=\sup\limits_{\tau\in\mathcal{T}}\E\left[e^{-(r-\gamma)\tau}\left(K-\widetilde{S}_{\tau}\right)^+|\ \widetilde{S}_0=s\right]\nonumber\\
=&\sup\limits_{\tau\in\mathcal{T}}\E\left[e^{-q\tau}\left(K-\widetilde{S}_{\tau}\right)^+|\ \widetilde{S}_0=s\right],\label{Aput}
\end{align}
where $\widetilde{S}_t=e^{-\gamma t}\bar{S}_t=e^{\widetilde{X}_t}$ for $\widetilde{X}_t=\bar{X}_t-\gamma t$ and $q=r-\gamma$ may appear to be negative.

In the context of real options, 
\cite{Battauz, DeDonno} consider {\it an optimal investment problem}:  a firm must decide when to invest in a single project and both the value of the project and the cost of entering it are stochastic.
This problem is reduced to the valuation of an American put option where the underlying is the ratio between the stochastic cost and the present value of the project (cost-to-value ratio) and the strike price is equal to 1. The interest rate is negative if the  growth rate of the project's present value dominates the discount rate.
When the cost of the investment is constant it is always convenient for the firm to wait and early exercise never occurs (see, for instance, \cite{DixitPindick}), thus this case is usually neglected. On the other hand, it turns out that the presence of a stochastic cost makes the problem relevant, because it may imply the existence of a well-defined exercise region.
As a last example where a negative interest rate may appear, we recall the presence nowadays of markets with a domestic non-positive short rate (as the Euro or the Yen denominated market):  in \cite{BattauzQuanto},  American quanto options  (put or call options written on a foreign security) are analyzed in such markets.

In \cite{Xia&Zhou} and   \cite{Battauz, DeDonno,  BattauzQuanto}, the question of the pricing and the identification of the exercise region for American options with  negative interest rate has been extensively studied under the assumption that the price of the underlying evolves according to a geometric Brownian motion.
In particular, 
\cite{Battauz, DeDonno,  BattauzQuanto} show that when the coefficients satisfy some special conditions,  a nonstandard double continuation region appears: exercise is optimally postponed not only when the option is not sufficiently in the money but also when the option is too deep in the money.
They explicitly characterize the value function and the two critical prices  which delimit the exercise region in the perpetual case, and study the properties of the time-dependent boundaries in the finite-maturity case.
A double continuation region can also appear with positive interest rate in the case of capped options. See
\cite{BroadieandDetemple} for the case of a capped option with growing cap, and \cite{DetempleandKitapbayev}
for the case of a capped option with two-level cap.

{\it Our aim} is to extend this analysis to {\it a L\'evy market}. It is well-known that the dynamics of securities in a financial market is described more accurately by processes with jumps.
Indeed, several empirical studies show that the log-prices of stocks have a  heavier left tail than   the normal distribution, on which the seminal Black-Scholes model is founded.
The introduction of jumps in a financial market dates back to 
\cite{Merton}, who added a compound Poisson process to the standard Brownian motion to better describe the dynamics of the logarithm of an asset. Since then, many papers have been written about the use of  general L\'evy processes in the modeling of financial markets and in the pricing of derivatives (see for instance \cite{Cont, Schoutens}).
Among the most recent examples of L\'evy processes used in modeling the evolution of the stock price process we refer to the normal inverse Gaussian model of
\cite{B10}, the hyperbolic
model of 
\cite{B42}, the variance gamma model of
\cite{B80}, the CGMY model of 
\cite{B24}, and the tempered stable process first introduced
by 
\cite{B68} and extended by
\cite{boyarchenkolevendorskii}.

American options in L\'evy markets have been studied in many papers.
\cite{Aase2010} studies perpetual put options and characterizes the continuation region in a jump-diffusion model. In a general L\'evy market, 
\cite{Mordecki} obtains closed formulas for the price of the perpetual call and put options and their critical prices, in terms of the law of the supremum and the infimum of a L\'evy process.  
\cite{boyarchenkolevendorskii}
exploit the Wiener-Hopf factorization and analytical methods to derive closed formulas for a large class of L\'evy processes and find explicit expressions in some special cases.
\cite{AsmussenAvramPistorius} also use the Wiener-Hopf factorization to compute the price of American put options but concentrate on L\'evy processes with two-sided phase-type jumps.
\cite{Erik1} use fluctuation theory to give an alternative proof of Mordecki's result;
see also 
\cite{Erik2}, and 
\cite{AliliKyprianou}, and references therein.

A major technique that has been widely used in the theory of optimal stopping
problems driven by diffusion processes is the free boundary formulation for the value
function and the optimal boundary.  The free boundary formulation consists primarily of a partial differential equation and (among other boundary conditions) the continuous
and smooth pasting conditions used to determine the unknown boundary and
specify the value function; see \cite{Peskir}.
In the context of the L\'evy market,
\cite{Budhiphd} gives sufficient and
necessary conditions for the continuous and smooth pasting conditions  (see also \cite{LambertonMikou}): in particular, the assumption of non zero volatility turns out to be fundamental.
A special case when the diffusion of the underlying process is degenerate has been also considered e.g. in \cite{Alvarez, boyarchenkolevendorskii, Cont, Kyprianou2007},\\
\cite{LambertonMikou}.
Another particular case of L\'evy market with nonzero volatility is the jump-diffusion market.
For this scenario, both the  methods of  variational inequalities and viscosity solutions of the boundary value problem have been used to study  American options; see
e.g. \cite{LambertonMikou, Pham, PhamViscosity}. We also refer to section 7 in 
\cite{Detemple2014},  for a general survey on American option in  the jump-diffusion model and other references.

In this paper {\it we propose a new approach to the pricing of  American options}.
To illustrate it, let us focus on  the put option. The price of the perpetual option
is a convex function of  the initial value $\widetilde{S}_0$ of the underlying  $\widetilde{S}_t$, and,  when $\widetilde{S}_0$ tends to 0, the price goes to $K$ if the interest rate is non-negative, otherwise goes to infinity. Moreover it always dominates the payoff function and coincides with it on the exercise region.
Therefore, if it is optimal to early exercise,  we can derive either a single continuation region (when the interest rate is non-negative, the  stopping region is a half-line)
 or a double continuation region (when the interest rate is negative, hence the stopping region is an interval).
Our approach consists in  considering all possible pairs of critical prices which may delimit the exercise region and maximize the value function over them: so doing, we
derive necessary and sufficient conditions for the existence of a double continuation region.

To compute the {\it two critical prices} which delimit the stopping region, we need to calculate {\it the Laplace transform of the entrance time of a closed interval}.
This is {\it a non-standard problem in the context of L\'evy processes} because they may jump over the interval without entering it. Still, we manage to solve this problem
for one-sided L\'evy processes, that is, for L\'evy processes either without positive jumps (so-called spectrally negative L\'evy processes)
or without negative jumps (so-called spectrally positive L\'evy processes).

The assumption of one-sided jumps in a L\'evy market  is quite common in financial modeling.
It  appears, for example, in
\cite{Erik1}, 
\cite{Erik2}, 
\cite{Budhiphd},
\cite{AliliKyprianou},
\cite{Avrampreprint}, and
\cite{Chan}. Moreover, it can be also found in
\cite{GerberShiu}, who exploit it to compute prices and optimal exercise strategies for the perpetual American put option in a jump-diffusion model, and in
\cite{ChesneyJeanblanc}, who evaluate  currency American options.
Similarly, 
\cite{BarrieuBellamy} analyze an investment problem in a jump-diffusion framework through a real option approach: they formulate it as an American call option problem on the ratio between the project value and the cost, which is assumed to jump only downwards, to represent a market crisis.

A core instrument in the fluctuation theory of completely asymmetric L\'evy processes are the so-called scale functions, which are usually exploited in the study of exit problems. For our entrance problem,  we derive  equations determining both ends of the stopping interval and the price of the American option in terms of these functions. Later, we analyze in details two cases, the Black-Scholes market and the jump-diffusion market with exponential jumps, where explicit solution and an extensive numerical analysis can be conducted.

We underline that the new method presented in this paper does not require {\it the smooth-fit condition}.
Still, {\it we prove that it holds} in our case. The most general conditions for smooth fit follow from
the arguments given in \cite[Thms. 4.1, 5.1 and 5.2]{LambertonMikou}
which are based on the variational inequality approach.
Under the more restrictive assumption that volatility is strictly positive ($\sigma >0$)
we give a new proof of this fact based on the extended version of It\^o's formula for L\'evy processes
that involves a local time on a curve separating two regions to
account for possible jumps over it; see
\cite {Eis},
\cite{Eisenbaum&Kyprianou}, 
\cite{Peskirlocaltime}, 
\cite{Budhiphd}.
Note that this approach is different from the one used in most of the papers, where
this condition is used to determine the unknown boundary and is also used as
a 'rule of thumb'; see 
\cite[p. 49]{Peskir}. It appears that the assumption of regularity of $0$ for $(-\infty,0)$ is crucial here; indeed
\cite{Budhiphd} and 
\cite{LambertonMikou} show that without this assumption the smooth-fit condition is not valid
in an asymmetric L\'evy market for $q\geq 0$.

The above results for American put options can be transformed into results for American call options using the so-called put-call symmetry formula. Another {\it contribution of this paper} is to show
that indeed {\it the put-call symmetry formula holds} in our, L\'evy framework with negative discount rate, case, and a dual relation between the two free boundaries of the call and the put options can be derived.
Our finding supplements 
\cite{FajardoMordecki} and, independently,
\cite{EberleinPapantaleon} who extend
to the L\'evy market the findings by 
\cite{CarrChesney}, under the assumption of a non-negative interest rate.
An analogous result for the negative discounting rate case is  obtained in  \cite{DeDonno} in the Black and Scholes market.
An interesting remark is due here: the symmetry formula shows that for some set of parameters one can
observe the double continuation region also for the American call option. However, there are also some cases when a single continuation region appears, in spite of the assumption of  a negative interest rate $q$. This occurs when the dividend rate $\delta$ is strictly positive, which is the case considered for instance by
\cite{Xia&Zhou}.  In Section \ref{sec:examples}, Tables \ref{Table1} and  \ref{Table2}, we collect all the possible cases for the put and call option in the Black and Scholes market to give a complete picture.

A further result of this paper is the generalization of the above results to {\it multiple stopping problems of Swing type}, that is
when successive exercise opportunities are possible and are separated by i.i.d. random refraction times.
These multiple stopping rules might appear when the same option can be exercised repeatedly in the future, meaning that an
investor can acquire a series of stock loans, or a firm can make an investment sequentially over time.
The features of refraction periods and multiple exercises
also arise in the pricing of Swing options commonly used for energy delivery. For instance, 
\cite{Carmona&Touzi}
formulate the valuation of a Swing put option as an optimal multiple stopping problem, with constant refraction periods,
under the geometric Brownian motion model. In a related study, 
\cite{Zeghal&Mnif} value a perpetual American
Swing put when the underlying L\'evy price process has no negative jumps.
\cite{Leung&Yamazaki&Zhang} consider call Swing options, under the assumption of a negative effective discount rate, and, as in the case of
\cite{Xia&Zhou}, they derive single continuation regions. Again this is a consequence of Assumption 2.1
made in \cite{Leung&Yamazaki&Zhang} that corresponds to the case when $q<0$ and $\delta >0$ in Table \ref{Table2} given in Section \ref{sec:examples}.
Therefore, in contrast to \cite{Leung&Yamazaki&Zhang},
we describe a complementary case of double continuation regions.

The paper is organized as follows.
In Section \ref{sec:model} we introduce basic facts concerning asymmetric, and in particular,
spectrally negative, L\'evy markets.
In Section \ref{sec:valuation} we compute  the price of
the American put option for completely asymmetric log L\'evy prices
and identify possibly double continuation region in this case.
Section \ref{sec:swing} deals with the Swing put options.
Section \ref{sec:putcall} presents a put-call symmetry formula for a L\'evy market with a  negative discounting rate, which allows for the extension of  the previous results to call options and Swing call options.
Finally, in Section \ref{sec:examples} we give an extensive numerical analysis
and in Section \ref{sec:conclusions} we summarize our results and propose possible generalizations.

\bigskip
\section{L\'evy-type financial market}\label{sec:model}


We  model the logarithm of  the underlying asset price by a completely asymmetric L\'evy process, that is
a stationary stochastic process with independent increments, right-continuous paths, with left limits for which all jumps have the same sign. Throughout the paper, we will use the following notation: because we mainly consider a spectrally negative L\'evy process,  that is a process with only downward jumps, we denote it simply with $X_t$; a spectrally positive process will be denoted by  $\widehat{X}_t$ (note that $\widehat{X}_t=-X_t$ where $X_t$ is spectrally negative); lastly, we denote by $\widetilde X_t$ a generic asymmetric L\'evy process.
More generally, when using the basic notation we refer to spectrally negative processes; a "hat" means that we are dealing with spectrally positive processes and a "tilde" with general asymmetric processes. So, for instance,
the process
  \begin{equation*}\label{asset1}S_t:=\exp\{X_t\}\end{equation*}
is the geometric L\'evy process, which describes the price of the underlying with $X$ spectrally negative.
 The asset price process corresponding to the spectrally positive case is
$\widehat{S}_t:=\exp\{\widehat{X}_t\}$,
 whereas
$\widetilde{S}_t:=\exp\{\widetilde{X}_t\}$ equals $S_t$ or $\widehat{S}_t$ depending on the scenario of the support of jumps.

We start from defining the  Laplace exponent of a completely asymmetric L\'evy process $\widetilde{X}_t$
\begin{equation*}\label{psi}
\widetilde{\psi} (\phi ):=\log\mathbb{E}[e^{
\phi \widetilde{X}_1}]
\end{equation*}
which is well defined for $\phi\geq 0$ for the spectrally negative case due to the non-positive jumps
and for $\phi\leq 0$ for the spectrally positive case due to the non-negative jumps.
By the L\'evy-Khintchine theorem, for $\mu\in\mathbb{R}$, $\sigma\geq 0$ and a L\'evy measure $\Pi$ defined on $\mathbb{R}\backslash\{0\}$ such that
\begin{align}
\int_{\mathbb{R}\backslash\{0\}}\left(1\wedge u^2\right)\Pi(\diff u)\nn < \infty
\end{align}
we have
\begin{equation}\label{eq:exponent}
\widetilde{\psi}(\phi)=\mu\phi
+\frac{1}{2}\sigma^{2}\phi^{2}+\int_{\mathbb{R}\backslash\{0\}}\big(\mathrm e^{\phi
u}-1-\phi u\mathbbm{1}_{(|u|<1)}\big)\Pi(\diff u).
\end{equation}
It is easily shown that $\widetilde{\psi}$ is
zero at the origin, tends to infinity at infinity and it is strictly
convex. We denote
\begin{equation*}
\widetilde{\Phi}(q):=\sup\{\phi>A:\widetilde{\psi}(\phi)=q\} \quad \textrm{and} \quad
\widetilde{\psi}(\widetilde{\Phi}(q))=q,
\end{equation*}
where $A$
is the last (going from the left hand side) asymptote of $\widetilde \psi(\phi) $.

\noindent {\it We will choose $q$ for which $\widetilde{\Phi}(q)$ is well-defined.}
 An important tool in the theory of spectrally negative L\'evy processes, which we will exploit to express our results, is the so-called scale function.
For a spectrally negative process $X_t$, with Laplace exponent $\psi$,
we define the first scale function as the unique continuous and strictly increasing function $W^{(q)}$ on $[0,\infty )$  with the following Laplace transform:
\begin{align}
\int_0^\infty e^{-\phi u}W^{(q)}(u)\diff u=\frac{1}{\psi (\phi )-q}\qquad\textrm{for }\phi>\Phi(q),\nonumber
\end{align}
where $\psi$ is a Laplace exponent of $X_t$.
The classical definition for scale function $W^{(q)}$ is given for $q\geq 0$. Note that this definition may be extended to $\mathbb{C}$ (see
\cite[Lemma 8.3]{KIntr}).
In particular, we can consider the case of $q<0$.
The function $W^{(q)}$ has left and right derivatives
on $(0,\infty)$.
Throughout
the paper we assume that X has unbounded variation or the jump measure is absolutely continuous
with respect to the Lebesgue measure.
In
that case,
\begin{align}
W^{(q)}\in \mathcal{C}^1(\mathbb{R}_+);\label{sigma0}
\end{align}
see 
\cite[Lem. 2.4, p. 117]{kyprianou}.


From now on, we further assume that the underlying process $\widetilde{S}_t$ is calculated under the martingale measure $\P$ (see \cite[Prop. 9.9]{Cont}).
Under this measure the discounted price process is a martingale and by definition of the Laplace exponent the process $\exp\{\phi \widetilde{X}_t-\widetilde{\psi}(\phi)t\}$
is also a martingale. Thus, under $\P$, we have
\begin{align}
\widetilde{\psi}(1)=q-\delta\label{esschertransform}
\end{align}
for $\widetilde{\psi}$ given in \eqref{eq:exponent}, $q$ being a discounting rate and $\delta$ being a dividend payment intensity or, if negative, borrowing rate or just the
costs of storing depending on applications.

Lastly, because we will use identities related to the first passage times above and below some level, we give the following definitions:
\begin{align}
\tau^+_a:=\inf\{t\geq 0 : X_t\geq a\},\qquad
\tau^-_a:=\inf\{t\geq 0 : X_t< a\}.\label{exittimes}
\end{align}
We will denote by  $\mathcal{T}$  the family of all stopping times  with respect to the right-continuous augmentation $\{\mathcal{F}_t\}_{\{t\geq 0\}}$ of the natural filtration of $X_t$ (or $\widehat{X}_t$
depending on the analyzed case) satisfying the usual conditions.

We conclude this preliminary section introducing the following notational convention:
\begin{align*}\mathbb{E}[\cdot; A]:=\mathbb{E}[\cdot\; \mathbbm{1}_{\{A\}}] \qquad \qquad \mbox{ for any event } A;\end{align*}
\begin{align*}\mathbb{P}_{(x)}\left[\cdot\right]:=\mathbb{P}\left[\cdot|X_0=x\right]=\mathbb{P}\left[\cdot|S_0=e^x\right];
\qquad
\mbox{ with } \mathbb{P}=\mathbb{P}_{(0)};
\end{align*}
\begin{align*}\mathbb{P}_{s}\left[\cdot\right]:=\mathbb{P}\left[\cdot|S_0=s\right];\end{align*}
$\E$, $\E_{(x)}$, $\E_s$ are the corresponding expectations to the above measures.

\bigskip
\section{American put option}\label{sec:valuation}

We start from the evaluation of the perpetual American put option with  negative discounting rate for both  the spectrally negative and spectrally positive cases. The problem for the spectrally negative case is defined in \eqref{Aput};
for the spectrally positive case, according to our notation, we will write
\begin{align*}
\widehat{V}_p(s):=\widehat{V}(s):=\sup\limits_{\tau\in\mathcal{T}}\E_s\left[e^{-q\tau}\left(K-\widehat{S}_{\tau}\right)^+\right].
\end{align*}

To make the problem of pricing  well-defined we assume throughout this paper that
\begin{align}
V(s)<\infty\qquad\text{and}\qquad \widehat{V}(s)<\infty.\label{condition1}
\end{align}
The next lemma gives a sufficient condition for \eqref{condition1} to hold true.
\begin{Lem}\label{lem:condition}
If $\E\left[X_1\right]$ is defined and valued in $(0,\infty)$, then the first condition of \eqref{condition1} is satisfied.
If $\E\left[X_1\right]$ is defined and valued in $(-\infty, 0)$, then the second condition of \eqref{condition1} is satisfied.
\end{Lem}
\begin{proof}
We prove only the spectrally negative case. The proof of the spectrally positive case is very similar.
According to \cite[Theorem 7.2]{KIntr}, if $\E\left[X_1\right]>0$, then
\begin{align}
\lim\limits_{t\uparrow\infty}X_t=\infty.\nn
\end{align}
As $S_t$ is geometric L\'evy process, then the above limit implies that $S_t$ also tends to infinity as $t$ goes to infinity.
Thus, the payoff at $\tau=\infty$ equals 0 and it is not optimal to stop there anymore. Now, as $\tau<\tau_{\rm last}(\log K)$ for
$\tau_{\rm last}(z):=\sup\{t\geq 0: X_t\leq z\}$
the discounting factor $\E_{(x)}e^{-q\tau}$ is finite. Indeed, by 
\cite{KIntr} we have
$$\E_{(x)}e^{-q\tau}\leq \E_{(x)}e^{-q\tau_{\rm last}(\log K)}\leq  \E_{(x)}e^{-q\tau_{\log K}^+}\E_{(0)}e^{-q\tau_{\rm last}(0)}
=e^{-\Phi(q) (x-\log K)}\E_{(0)}e^{-q\tau_{\rm last}(0)}.$$
Now, by 
\cite[Thm 2]{Eriklast}
we know that $\E_{(0)}e^{-q\tau_{\rm last}(0)}$ is finite as long as $\Phi(q)$ is well-defined which we assumed in Section \ref{sec:model}.
This observation together with the boundness of the payoff function $(K-s)^+$ give the condition \eqref{condition1}.
\end{proof}

Following \cite{Battauz} which deals with the Black-Scholes model, that is when $\widetilde{X}_t $ is an arithmetic Brownian motion,
we will show that in the completely asymmetric L\'evy  market the optimal stopping rule $\tau\in \mathcal{T}$ has the following interval-type form:
\begin{align}
\tau_{l,u}:=\inf\{t\geq 0:\ \widetilde{X}_t\in[l,u]\}=\inf\{t\geq 0:\ \widetilde{S}_t\in[\mathbf{l},\mathbf{u}]\},\label{tau}
\end{align}
for some $\mathbf{l}:=e^l$, $\mathbf{u}:=e^u$ and $\mathbf{l}\leq \mathbf{u}\leq K$, namely it is an \emph{entrance time}.
This means that an investor should exercise the option whenever the price process $S_t$ enters the interval $[\mathbf{l}, \mathbf{u}]$.
One can give heuristic arguments for the above choice of the optimal stopping time.
It is never optimal to exercise early if $\widetilde{S}_t\geq K$ because the payoff function is then equal to 0. Thus one would wait until the process falls below some level $\mathbf{u}\leq K$. On the other side, if $q<0$ then for $S_0=0$ the value of the American put option dominates the payoff function, because, when $q<0$,
$\sup\limits_{\tau \in \mathcal{T}}\E_0\left[e^{-q\tau}(K-0)^+\right]=K\sup\limits_{\tau \in \mathcal{T}}\E_0\left[e^{-q\tau}\right]>K$. This means that the price of the underlying asset cannot be too small. Thus, if it is optimal to exercise early, then the stopping region is $[\mathbf{l}^*, \mathbf{u}^*]$ for some optimal levels $\mathbf{l}^*$ and $\mathbf{u}^*$.

Formally, the structure \eqref{tau} of the stopping region follows
from the following fact which can be of own interest.

\begin{Lem}\label{newlemma}
The value function is non-increasing, convex (and hence continuous) on the open set $(0,+\infty)$ and the optimal stopping rule has the form  \eqref{tau}.
\end{Lem}
\begin{proof}
Because the payoff function of the put option is continuous,   by
\cite[Thm. 2.2, p. 29, Cor. 2.9, p. 46, Rem. 2.10, p. 48 and (2.2.80), p. 49]{Peskir}
the optimal stopping rule is given by
\begin{equation}\label{optstop0}
\tau^*:=\inf\{t\geq 0: \widetilde{V}(\widetilde{S}_t)=(K-\widetilde{S}_t)^+\}.
\end{equation}
To identify its special form,
observe that the European option price $\E_s\left[e^{-qt}(K-S_t)^+\right]$
is strictly convex in the underlying asset $s$.
Indeed, first observe that from its definition this price is continuous. Then from classical mollifying arguments without loss of generality we can
assume that this price is also differentiable. In the last step we follow \cite{Bergmanetal} heavily using fact that log price in our case depends linearly
on the initial position of the process $\widetilde{X}$.
Now, following the proof of 
\cite[Cor. 2.6]{Ekstrom} (taking the maturity tending to infinity at the last step),
we can conclude that $\widetilde{V}(s)$ is also convex. By definition $\widetilde{V}(s)$ is also non-increasing.
Besides it, $\widetilde{V}(s)$ can not smoothly paste to zero for $s>K$. Indeed, by our assumptions
we have $\widetilde{V}(s)>0$ for any $s>K$ because
\[\widetilde{V}(s)\geq \E_s\left[e^{-q\tau_{l,u}}(K-\widetilde{S}_{\tau_{l,u}})^+\right]\geq (K-u) \E_s\left[e^{-q\tau_{l,u}}\right] >0.\]
This implies that the function $\widetilde{V}$ and the payoff function $K-s$ can cross each other at most two times.
In other words, the stopping region has to be an interval or a half-line.
\end{proof}
\begin{Rem}
In this paper we decided to deal only with the infinite time horizon case but part of the results hold true for the finite maturity case as well.
For example, defining
$$V((S(t),t) = \mbox{ess}\sup\limits_{\tau\in\mathcal{T}_{t,T}} \E\left[e^{-q(\tau - t)}\left(K-\widetilde{S}_{\tau}\right)^+|\ \mathcal{F}_t\right] $$
for maturity $T$, where $\mathcal{T}_{t,T}$ is the set of stopping times satisfying   $t \leq \tau \leq T$,
following 
\cite{Ekstrom} one can show that $V(s,t)$ is non increasing and convex with respect  to $s$ (see also \cite[Prop. 2.2]{LambertonMikou}.
As $(K-s)^+ \leq V(s,t) \leq V(s)$  and $V(0,t) = e^{-q(T-t)} K > K $  for negative $q$,  if the infinite maturity option has an exercise  region of the form $[\textbf{l}^*, \textbf{u}^*]$, then $V(s,t)$ has an exercise region (depending on $t$) of the form $[\textbf{l}(t), \textbf{u}(t)]$ with $0\leq \textbf{l}(t) \leq \textbf{l}^* \leq \textbf{u}^* \leq \textbf{u}(t) \leq K$.
\end{Rem}

If, as in our case, $q<0$, then the function $\widetilde{V}$ is strictly greater than $K$ when $s$ is close to 0, therefore the stopping region is necessarily an interval.
To identify the optimal levels $\mathbf{l}^*$ and $\mathbf{u}^*$, we have to calculate the value function
for the stopping rule \eqref{tau}. Note that these arguments hold for any L\'evy process.
The additional assumption of one-sided jumps allows for explicitly computing the value function.
Let
\begin{align*}
v(x,l,u):=\E_{(x)}\left[e^{-q\tau_{l,u}}\left(K-e^{X_{\tau_{l,u}}}\right)^+\right], \qquad \widehat{v}(x,l,u):=\E_{(x)}\left[e^{-q\tau_{l,u}}\left(K-e^{\widehat{X}_{\tau_{l,u}}}\right)^+\right]
\end{align*}
for $$x:=\log s.$$

In the following lemma we  express the function $v$ and $\widehat{v}$ in terms of the scale function $W^{(q)}$
of the process $X$ for the spectrally negative case and of the process $-\widehat{X}$ for the spectrally positive case.
\begin{Lem}\label{v_value}
Let $X_t$ be a spectrally negative L\'evy process. Assume that $\Phi(q)$ is well-defined and $\Phi(q)<0$.

(i) Then the function $v$ is given by
\begin{align}
v(x,l,u)=&\left(K-e^x\right)\mathbbm{1}_{(x\in [l,u])}+\left(K-e^l\right)e^{-\Phi(q)(l-x)}\mathbbm{1}_{(x<l)}\label{eq:v_value}\\
&+\Bigg\{\int_0^\i\int_{(0,\i)} e^{-\Phi(q)(l-u+y)\vee 0}\left(K-e^{l\vee (u-y)}\right)r^{(q)}(x-u,z)\Pi(-z-\diff y)\diff z\nn\\
&\qquad+\left(K-e^u\right)\frac{\sigma^2}{2}\left(\Wp(x-u)-\Phi(q)W^{(q)}(x-u)\right)\Bigg\}\mathbbm{1}_{(x>u)}\nn
\end{align}
where
\begin{align*}
r^{(q)}(x,z)=e^{-\Phi(q)z}W^{(q)}(x)-W^{(q)}(x-z)
\end{align*}
is a resolvent density of the process $X_t$ killed on exiting $[0,\i)$.

(ii) Then for the spectrally positive L\'evy process $\widehat{X}_t=-X_t$ we have
\begin{align*}
\widehat{v}(x,l,u)=&\left(K-e^x\right)\mathbbm{1}_{(x\in [l,u])}+\left(K-e^u\right)e^{-\Phi(q)(x-u)}\mathbbm{1}_{(x>u)}
\\
&+\Bigg\{\int_0^\i\int_{(0,\i)} e^{-\Phi(q)(l-u+y)\vee 0}\left(K-e^{(l+y)\wedge u}\right)r^{(q)}(l-x,z)\Pi(-z-\diff y)\diff z\nn\\
&\qquad+\left(K-e^l\right)\frac{\sigma^2}{2}\left(W^{(q)\prime}(l-x)-\Phi(q)W^{(q)}(l-x)\right)\Bigg\}\mathbbm{1}_{(x<l)}.\nn
\end{align*}
\end{Lem}
\begin{proof}
To prove the first case (i) note that we can consider only $u<\log K$. Indeed, if $u\geq\log K$, then the payoff is zero because $(K-e^{u})^+=0$.
This cannot produce the optimal value since waiting longer gives a nonzero payoff.
With this assumption we may rewrite the value function as follows
\begin{align}
v(x,l,u):=\E_{(x)}\left[e^{-q\tau_{l,u}}\left(K-e^{X_{\tau_{l,u}}}\right)\right].\nn
\end{align}
As the process $X_t$ is spectrally negative, when it starts  below $l$, it enters the interval $[l,u]$ in a continuous way and hence $X_{\tau_{l,u}}=l$. Thus for $x<l$ we get
\begin{align}
v(x,l,u)=\E_{(x)}\left[e^{-q\tau^+_l};\ \tau^+_l<\i\right]\left(K-e^l\right)=e^{-\Phi(q)(l-x)}\left(K-e^l\right).\nn
\end{align}
The first expectation given above is the solution to the so-called one-sided exit problem and is a well known identity for L\'evy processes (see
\cite[Theorem 8.1]{KIntr}).

If $x>u$, then there are two cases: either  the underlying process $X$ enters the interval $[l,u]$ going downward or  the process $X$ jumps from $(u,\i)$ to $(-\i,l)$ and after that enters $[l,u]$, that is, it enters $[l,u]$ creeping upward. Note that it is impossible that $X$ jumps from $[u,\i)$ to $(-\i,l]$ and after that goes back to $[u,\i)$
because it has only downward jumps. We may separate this two cases in the function as follows
\begin{align}\label{eq_v:cases}
v(x,l,u)=\E_{(x)}\left[e^{-q\tau_{l,u}}(K-e^{X_{\tau_{l,u}}});\ \tau^-_u<\tau_l^-\right]+\E_{(x)}\left[e^{-q\tau_{l,u}}(K-e^{X_{\tau_{l,u}}});\ \tau^-_u=\tau_l^-\right]
\end{align}
The first term in the right-hand side equation corresponds to the case when the process $X$  enters the interval $[l,u]$ through jump or by creeping downwards to level $u$. This last case
 is possible when
 $\sigma >0$ (see 
 \cite[Exercise 7.6]{KIntr}).
 Thus
\begin{align}
&\E_{(x)}\left[e^{-q\tau_{l,u}}(K-e^{X_{\tau_{l,u}}});\ \tau^-_u<\tau_l^-\right]=\E_{(x)}\left[e^{-q\tau_u^-}(K-e^{X_{\tau_{u}^-}});\ X_{\tau^-_u}\in [l,u]\right]\nn\\
&\quad = \int_{(l,u)}(K-e^s)\E_{(x)}\left[e^{-q\tau_u^-};\ X_{\tau^-_u}\in \diff s\right]+\left(K-e^u\right)\E_{(x)}\left[e^{-q\tau_u^-};\ X_{\tau^-_u}=u\right]\nn\\
&\quad = \int_{(0,u-l)}\left(K-e^{u-y}\right)\E_{(x-u)}\left[e^{-q\tau_0^-};\ -X_{\tau_0^-}\in\diff y\right]+\left(K-e^u\right)\E_{(x-u)}\left[e^{-q\tau_0^-};\ X_{\tau_0^-}=0\right].\nn
\end{align}
The joint law for the first passage below and the position at this time are given in \cite[Chapter 8.4]{KIntr}
\begin{align}
&\E_{(x)}\left[e^{-q\tau_0^-};\ -X_{\tau^-_0}\in \diff y\right]=\int_0^\infty \left[ e^{-\Phi(q)z}W^{(q)}(x)-W^{(q)}(x-z)\right] \, \Pi(-z-\diff y)\diff z,\label{applying}\\
&\E_{(x)}\left[e^{-q\tau_0^-};\ X_{\tau^-_0}=0\right]=\frac{\sigma^2}{2}\left[W^{\prime (q)}(x)-\Phi(q)W^{(q)}(x)\right]\nn.
\end{align}
The second formula above is well-defined because
by \eqref{sigma0} the derivative of the scale function $W^{\prime(q)}$ exists.

The second term in the right-hand side of equation \eqref{eq_v:cases} refers to the case where the process $X$ firstly jumps from $(u,\i)$ to $(-\i,l)$ and, after that, enters $[l,u]$. Note that when the process jumps below level $l$, it creeps upward  starting from $X_{\tau^-_l}<l$. Hence
\begin{align}
&\E_{(x)}\left[e^{-r\tau_{l,u}}(K-e^{X_{\tau_{l,u}}});\ \tau^-_u=\tau_l^-\right]=\E_{(x)}\left[e^{-r\tau_{l,u}}(K-e^{X_{\tau_{l,u}}});\ X_{\tau^-_u}<l\right]\nn\\
&\quad =\E_x\left[e^{-r\tau_u^-}\E_{(X_{\tau^-_u})}\left[e^{-r\tau_{l,u}}(K-e^{X_{\tau_{l,u}}})\right];\ X_{\tau^-_u}<l\right]\nn\\
&\quad=\int_{(u-l,\i )}\E_{(x-u)}\left[e^{-q\tau_0^-}\E_{(u-y)}\left[e^{-q\tau_{l,u}}\left(K-e^{X_{\tau_{l,u}}}\right)\right]; -X_{\tau_0^-}\in\diff y\right]\nn\\
&\quad =\int_{(u-l,\i)}e^{-\Phi(q)(l-u+y)}\left(K-e^{l}\right)\E_{(x-u)}\left[e^{-q\tau_0^-}; -X_{\tau_0^-}\in\diff y\right].\nn
\end{align}
Applying the joint law for $\left(\tau^-_u,\ X_{\tau^-_u}\right)$ given earlier in \eqref{applying} we get the assertion for $x>u$.

Obviously, when $x\in [l,u]$ then the first moment when process enters $[l,u]$ happens immediately, i.e. $\tau_{l,u}=0$ and the value is $K-e^x$.
This completes the proof of the part (i). The proof of the case (ii) is similar.
\end{proof}

We sum up our findings in the following theorem.
\begin{Thm}\label{thm:optimal}
Assume that the value functions $v(x,l,u)$ and $\widehat{v}(x,l,u)$   calculated in Lemma 3 for spectrally negative and spectrally positive
log asset L\'evy price process respectively admits unique maximizers $\mathbf{l}^*=e^{l^*}$ and $\mathbf{u^*}=e^{u^*}$. Then the stopping set  for the American put option with negative interest rate in a completely asymmetric L\'evy market is not empty and
the optimal stopping rule, which is finite,
is given in (3.2)
for  $\mathbf{l}^*=e^{l^*}$ and $\mathbf{u^*}=e^{u^*}$. The prices of the perpetual American option in the two cases are respectively
$$ V_p(s) = \sup_{l \leq  u  } v(\log s, l, u) = v(\log s, l^*, u^*), \qquad  \qquad \widehat{V}_p(s) = \sup_{ l \leq u  } \widehat{v}(\log s,l,u) = \widehat{v}(\log s,l^*,u^*).
$$

\end{Thm}

Traditionally, the optimal stopping regions are identified using the  \textit{continuous} and \textit{smooth fit} conditions. Though we do not need them to prove our results,  we show in the following proposition that in fact under mild assumptions the optimal levels $l^*$ and $u^*$ satisfy these conditions.
\begin{Prop}\label{smooth_fit} Assume that the optimal stopping rule is finite, that is, the stopping region is non-empty.
Consider the following smooth and continuous fit conditions:
\begin{alignat*}{2}
&\frac{\partial}{\partial x}\widetilde{v}(x,l^*,u^*)|_{x\uparrow l^*}=\frac{\partial}{\partial x}\widetilde{v}(x,l^*,u^*)|_{x\downarrow l^*}&&\quad\textrm{\textit{(smooth fit} at }l^*),
\\
&\frac{\partial}{\partial x}\widetilde{v}(x,l^*,u^*)|_{x\downarrow u^*}=\frac{\partial}{\partial x}\widetilde{v}(x,l^*,u^*)|_{x\uparrow u^*}&&\quad\textrm{\textit{(smooth fit} at }u^*),
\\
&\widetilde{v}(x,l^*,u^*)|_{x\uparrow l^*}=\widetilde{v}(x,l^*,u^*)|_{x\downarrow l^*}&&\quad\textrm{\textit{(continuous fit} at }l^*),
\\
&\widetilde{v}(x,l^*,u^*)|_{x\downarrow u^*}=\widetilde{v}(x,l^*,u^*)|_{x\uparrow u^*}&&\quad\textrm{\textit{(continuous fit} at }u^*).
\end{alignat*}

(i) The optimal value function $\widetilde{V}(x)=\widetilde{v}(x,l^*,u^*)$ of the perpetual American put option for a completely asymmetric L\'evy process always satisfies
continuous fit conditions.

(ii) If $0$ is regular for $(-\infty, 0)$, then the smooth-fit principle is satisfied for $u^*$.

(iii) If $0$ is regular for $(0, \infty)$, then the smooth-fit principle is satisfied for $l^*$.

(iv) If $\widetilde{X}_t$ is of bounded variation and $d:=\lim_{t\downarrow 0} \frac{\widetilde{X}_t}{t}>0$ then
the smooth-fit principle at $u^*$ is not satisfied.

(v) If $\widetilde{X}_t$ is of bounded variation and $d<0$ then
the smooth-fit principle at $l^*$ is not satisfied.
\end{Prop}
\begin{Rem}
The equivalent conditions of the regularity of $0$ for negative half-line in terms of the volatility $\sigma$ and the jump measure $\Pi$
are given, e.g,. in \cite[Prop. 7]{AliliKyprianou} and \cite[Prop. 4.1]{LambertonMikou}.
\end{Rem}

\begin{Rem}
We do not give any specific procedure for computing the optimal thresholds $l^*, u^*$. The main reason is that for a wild class of spectrally negative L\'evy process $X$ the scale function $W^{(q)}(x)$
is in the form of liner combination of some exponential functions. This is the case for example when jumps have phase-type  distribution. In this case the value functions $v(x,l,u)$ and $\widehat{v}(x,l,u)$ could be calculated explicitly.
Thus identifying optimal $l^*$ and $u^*$ is equivalent to finding the values $l$ and $u$ maximizing these functions. That could be done explicitly or   by any numerical procedure finding the maximum of the function. For instance, in our Examples in Section \ref{sec:examples}, first-order conditions can be invoked instead of the smooth-fit condition, to compute $l^*$ and $u^*$.
\end{Rem}

\begin{proof}
The first statement follows from the continuity of the value function proved in Lemma \ref{newlemma}.
Statement (ii) follows from very similar arguments like the ones given in the proof of \cite[Thms. 4.1 and 5.1]{LambertonMikou}.
The only explanation requires inequality appearing in the proof of \cite[Thms. 4.1]{LambertonMikou}
\[\widetilde{V}(u^*)\geq \E_{u^*+h} \left[e^{-q\tau_h}g\left(\widetilde{S}_{\tau_h}\right)\right]\]
for
\begin{equation}\label{payoff}
g(s):=(K-s)^+
\end{equation}
and $\tau_h:=\inf\{t\geq 0: \widetilde{S}_t<u^*\}$ for $h>0$.
This inequality is straightforward from definition of the value function because $\tau_h$ is
a stopping time since, by Lemma \ref{newlemma}, we have $\widetilde{V}(u^*)\geq \widetilde{V}(u^*+h)$.
The case of the lower critical point $l^*$ could be resolved in the same way by taking $l^*-h$ instead of $u^*+h$.
Finally the two last statements follow from the proof of \cite[Thm. 5.2]{LambertonMikou}. In particular, the inequality
\cite[(5.13)]{LambertonMikou} is obvious in this case (by taking the negative interest rate $q<0$ instead $r$ there).

If we assume a more stronger condition than the regularity of $0$ for negative or positive half-line, namely that
\begin{align}
\sigma >0\label{sigma}
\end{align}
then we can give a different and new proof of the smooth fit conditions at $l^*$ and $u^*$.
We also assume that the density of jump measure satisfies
\begin{equation}\label{density}
\pi(x)\leq C|x|^{-1-\alpha}
\end{equation}
in a neighborhood of the origin,
for some $0<\alpha<1$ and $C>0$.

We assume that $X_t$ is spectrally negative, hence $\widetilde{v}=v$. The proof for the spectrally positive case is the same.
Under condition \eqref{sigma} (recall also the global assumption \eqref{density}), we have
\begin{equation}\label{scalesmooth}
W^{(q)}\in \mathcal{C}^3(\mathbb{R}_+);
\end{equation}
see 
\cite[Thm. 3.11, p. 140]{kyprianou} and hence by Lemmas
\ref{newlemma} and \ref{v_value} we can conclude that
$v(x,l^*,u^*)\in \mathcal{C}^2(\mathbb{R}_+)$ for $x \neq l^*$ and $x\neq u^*$ because the optimal thresholds $l^*$ and $u^*$ do not depend on initial asset prices $s=e^x$.
Indeed, in this case
\begin{align}
\frac{\partial}{\partial x}v(x,l,u)=&-e^x\mathbbm{1}_{(x\in [l,u])}+\Phi(q)\left(K-e^l\right)e^{-\Phi(q)(l-x)}\mathbbm{1}_{(x<l)}\label{eq:v_value}\\
&+\Bigg\{\int_0^\i\int_{(0,\i)} e^{-\Phi(q)(l-u+y)\vee 0}\left(K-e^{l\vee (u-y)}\right)\frac{\partial}{\partial x}r^{(q)}(x-u,z)\Pi(-z-\diff y)\diff z\nn\\
&\qquad+\left(K-e^u\right)\frac{\sigma^2}{2}\left(W^{(q)\prime\prime}(x-u)-\Phi(q)W^{(q)\prime}(x-u)\right)\Bigg\}\mathbbm{1}_{(x>u)}\nn
\end{align}
for
\begin{align*}
\frac{\partial}{\partial x} r^{(q)}(x-u,z)=e^{-\Phi(q)z}W^{(q)\prime}(x-u)-W^{(q)\prime}(x-u-z)
\end{align*}
and
\begin{align}
\frac{\partial^2}{\partial^2 x}v(x,l,u)=&-e^x\mathbbm{1}_{(x\in [l,u])}+\Phi(q)^2\left(K-e^l\right)e^{-\Phi(q)(l-x)}\mathbbm{1}_{(x<l)}\label{eq:v_value}\\
&+\Bigg\{\int_0^\i\int_{(0,\i)} e^{-\Phi(q)(l-u+y)\vee 0}\left(K-e^{l\vee (u-y)}\right)\frac{\partial^2}{\partial^2 x}r^{(q)}(x-u,z)\Pi(-z-\diff y)\diff z\nn\\
&\qquad+\left(K-e^u\right)\frac{\sigma^2}{2}\left(W^{(q)\prime\prime\prime}(x-u)-\Phi(q)W^{(q)\prime\prime}(x-u)\right)\Bigg\}\mathbbm{1}_{(x>u)}\nn
\end{align}
for
\begin{align*}
\frac{\partial^2}{\partial^2 x} r^{(q)}(x-u,z)=e^{-\Phi(q)z}W^{(q)\prime\prime}(x-u)-W^{(q)\prime\prime}(x-u-z).
\end{align*}
Hence we can apply the change of variable formula.

From Theorem \ref{thm:optimal} it follows that $v(x,l^*,u^*)$ is the value function of the American put option
and from the general optimal stopping theory we know that $v(x,l^*,u^*)$ dominates the payoff function and $e^{-qt}v(X_t,l^*,u^*)$ is a supermartingale
(see  
\cite{Peskir}).

Let us focus on the boundary $u^*$. As the value function dominates the payoff function and both are non-increasing and convex, we have that
\begin{equation}\label{onsideineq}\frac{\partial}{\partial x}v(x,l^*,u^*)|_{x\downarrow u^*}\geq\frac{\partial}{\partial x}v(x,l^*,u^*)|_{x\uparrow u^*}.\end{equation}
To prove the inequality in the opposite direction we
use the change of variable formula presented in\\ 
\cite[p. 208]{Eisenbaum&Kyprianou} together with Dynkin formula
\begin{align}
e^{-qt}v(X_t,l^*,u^*)=&M_t+v(x,l^*,u^*)+\int_0^t(\mathcal{A}-q)v_1(X_s,l^*,u^*)\diff s +\int_0^t(\mathcal{A}-q)v_2(X_s,l^*,u^*)\diff s\nn\\
&+ \int_0^te^{-qs}\frac{\partial}{\partial x}\left(v(X_{s+},l^*,u^*)-v(X_{s-},l^*,u^*)\right)\diff L^{u^*}_s,\nn
\end{align}
where $M_t$ is a local martingale, $L^{u^*}_s$ is a local time of $X$ at $u^*$, $\mathcal{A}$ is an infinitesimal generator of $X$ and $v_1(x,l^*u^*)=v(x,l^*,u^*)|_{x>u^*}$ and $v_2(x,l^*u^*)=v(x,l^*,u^*)|_{x\leq u^*}$.
Note that by Lemma \ref{v_value} and by \eqref{scalesmooth}, the functions $v_1$ and $v_2$ are in the domain of the infinitesimal generator $\mathcal{A}$ (see
\cite[Thm. 2]{Eisenbaum&Kyprianou}).
Now, from the general theory of optimal stopping (see \cite{Peskir} for details) we know that $e^{-qt}v(X_t,l^*,u^*)$ is a supermartingale. Thus,
\begin{align}
\E_{(x)}\left\{\int_w^t(\mathcal{A}-q)v_1(X_s,l^*,u^*)\diff s+\int_w^t(\mathcal{A}-q)v_2(X_s,l^*,u^*)\diff s +\int_w^te^{-qs}\frac{\partial}{\partial x}\left(v(X_{s+},l^*,u^*)-v(X_{s-},l^*,u^*)\right)\diff L^{u^*}_s\right\}\leq 0\label{inequ1}
\end{align}
for any $0\leq w\leq t$.
From 
\cite[Thm. 3]{Eisenbaum&Kyprianou} it follows that the process
$$t\rightarrow \int_w^te^{-qs}\frac{\partial}{\partial x}\left(v(X_{s+},l^*,u^*)-v(X_{s-},l^*,u^*)\right)\diff L^{u^*}_s$$ is of unbounded variation on any finite interval
similarly as $X_t$ is by assumption
\eqref{sigma}.
On the other hand, the processes
$t\rightarrow \int_0^t(\mathcal{A}-q)v_1(X_s,l^*,u^*)\diff s $ and $t\rightarrow \int_0^t(\mathcal{A}-q)v_2(X_s,l^*,u^*)\diff s $ are of bounded variation.
Thus, taking $t\rightarrow w$ in~\eqref{inequ1} we can conclude that
\begin{align}\label{Uwt}
U(w,t):=\E_{(x)}\int_w^te^{-qs}\frac{\partial}{\partial x}\left(v(X_{s+},l^*,u^*)-v(X_{s-},l^*,u^*)\right)\diff L^{u^*}_s\leq 0
\end{align}
for all sufficiently small $w$ and $t$.
Indeed, assume \emph{a contrario} that for some $w<t$, we have $U(w,t)>0$.
As the local time $L_t^{u^*}$ is nondecreasing and it increases only when process $X_t$ enters the half-line $(u^*,\infty)$ leaving the set $(-\infty, u^*)$,
then, by taking $x=u^*$, from \eqref{onsideineq}
we have that $U(w,t)$ is equal to its variation which is $+\infty$. Then
\begin{align*}
&U(w,t)+\E_{(x)}\left\{\int_w^t(\mathcal{A}-q)v_1(X_s,l^*,u^*)\diff s+\int_w^t(\mathcal{A}-q)v_2(X_s,l^*,u^*)\diff s\right\}\\&\qquad  \geq
U(w,t) + (w-t)\E_{(x)}\inf_{s\in[w,t]}\left( (\mathcal{A}-q)v_1(X_s,l^*,u^*)+(\mathcal{A}-q)v_2(X_s,l^*,u^*)\right)=+\infty\end{align*}
which is a contradiction to \eqref{inequ1}.

From \eqref{Uwt} the following inequality must hold true
$$\frac{\partial}{\partial x}v(x,l^*,u^*)|_{x\downarrow u^*}-\frac{\partial}{\partial x}v(x,l^*,u^*)|_{x\uparrow u^*}\leq 0.$$
This inequality together with the already proved converse inequality $\frac{\partial}{\partial x}v(x,l^*,u^*)|_{x\downarrow u^*}\geq\frac{\partial}{\partial x}v(x,l^*,u^*)|_{x\uparrow u^*}$
completes the proof of the \textit{smooth fit} property at $u^*$. The smooth fit property at $l^*$ follows by similar arguments or by direct calculation.
\end{proof}

 \bigskip

\section{Swing put option}\label{sec:swing}


A multiple stopping option, the so-called \textit{Swing option}, is a derivative with $N\geq 1$ exercise opportunities with some refraction period between them. Let $\{\delta_i\}_{i=1}^{N-1}$ be the refraction periods between consecutive exercise times. We assume that $\delta_i$ are positive random variables and that
 $\widetilde X_{\delta_i}$
has no atoms. The value of the perpetual Swing put option is defined as follows
\begin{align}
V^{(N)}_p(s):=V^{(N)}(s):=\sup\limits_{\underline{\tau}\in\mathcal{T}^{(N)}}\E_s\left[\sum\limits_{i=1}^N e^{-q\tau_i}(K-\widetilde{S}_{\tau_i})^+;\ \tau_i<\infty\right],\label{swing}
\end{align}
for some strike price $K>0$, where
\begin{align}
\mathcal{T}^N:=\{\underline{\tau}=(\tau_1,...,\tau_N):\ \tau_1\in\mathcal{T}\textrm{ and }\tau_i=\tau_{i-1}+\delta_{i-1}+\zeta\circ\theta_{\tau_{i-1}+\delta_{i-1}}
\textrm{ for } \zeta \in \mathcal{T} \textrm{ and }i=2,...,N
\},\label{Tn}
\end{align}
where $\theta_t$ is a shift operator associated with the transition probability of $\widetilde{X}$.
Note that $V^{(1)}(s)$ is the value of the American put option analyzed in the previous section and that condition \eqref{condition1} implies that  $ V^{(N)}(s) < +\infty$.
To solve the optimal stopping problem \eqref{swing} of the Swing option we follow the approach of  
\cite{Leung&Yamazaki&Zhang} and define recursively   following  single optimal stopping problems:
\begin{align}
\widetilde{V}^{(k)}(s):=\sup\limits_{\tau\in\mathcal{T}}\E_s\left[e^{-q\tau}g^{(k)}(\widetilde{S}_\tau);\ \tau<\infty\right]\label{Vk}
\end{align}
with
\begin{align}
g^{(k)}(s):=g(s)+\E_s\left[e^{-q\delta_k}\widetilde{V}^{(k-1)}(\widetilde{S}_{\delta_k})\right]\label{recursion}
\end{align}
for $g(s)$ given in \eqref{payoff},
$k=1,..,N$ and $\widetilde{V}^{(0)}(s):=0$.
\begin{Lem}\label{lem:Vk_convex}
For every $k\in\{1,...,N\}$ and all $s\in[0,\infty)$ the function $\widetilde{V}^k(s)$ is a non-increasing and convex function of $s$.
\end{Lem}
\begin{proof}
We prove the result by induction.  For $k=1$ and the payoff function \eqref{payoff} the result follows immediately from the fact that the value of the American put option is
a non-increasing and convex function of $s$ (see e.g. 
\cite[Cor. 2.6]{Ekstrom}).
Now, suppose that the results holds for $k=l-1$ for some $l\in\{2,...,N-1\}$. Thus
\begin{align}
g^{(k)}(s):=g(s)+\E_s\left[e^{-q\delta_k}\widetilde{V}^{(k-1)}(\widetilde{S}_{\delta_k})\right]=g(s)+\E\left[e^{-q\delta_k}\widetilde{V}^{(k-1)}(s\widetilde{S}_{\delta_k})\right]\nn
\end{align}
is also non-increasing and convex with respect to $s$ as the functions $g(s)$ and $V^{(k-1)}$ are of the same type.
We can use the same arguments as the ones given in the proof of 
\cite[Cor. 2.6]{Ekstrom} to prove that
$\widetilde{V}^{(l)}(s)$ is a non-increasing and convex function of $s$.
Then, the claim is proved.
\end{proof}
From the previous considerations given in Section \ref{sec:valuation} we know that there exist $\mathbf{l}_1^*=\mathbf{l}^*=e^{l^*}$ and $\mathbf{u}_1^*=\mathbf{u}^*=e^{u^*}$ such that $\widetilde{V}^{(1)}(s)=g(s)$ for any $s\in[\mathbf{l}_1^*,\mathbf{u}_1^*]$.
We now extend this result for $k\in\{2,...,N\}$ adapting the idea of the proof of Lemma 2.3 of 
\cite{Carmona&Touzi}.
\begin{Lem}\label{lem:Vk_lu}
Let $0\leq \mathbf{l}_1^*\leq \mathbf{u}_1^*\leq K$. Then for all $k\in\{1,...,N\}$ and all $s\in[\mathbf{l}^*_1,\mathbf{u}^*_1]$, we have $g^{(k)}(s)=\widetilde{V}^{(k)}(s)$.
\end{Lem}
\begin{proof}
It is sufficient to show that $\widetilde{V}^{(k)}(s)\leq g^{(k)}(s)$ on $[\mathbf{l}^*_1,\mathbf{u}^*_1]$ as the reverse inequality is trivial. From the definition of $g^{(k)}$ we can write
\begin{align}
\widetilde{V}^{(k)}(s)&=\sup\limits_{\tau\in\mathcal{T}}\E_s\left[e^{-q\tau}g(\widetilde{S}_\tau)+e^{-q(\tau+\delta_1)}\widetilde{V}^{(k-1)}(\widetilde{S}_{\tau+\delta_1})\right]\leq \widetilde{V}^{(1)}(s)+\sup\limits_{\tau\in\mathcal{T}}\E_s\left[e^{-q(\tau+\delta_1)}\widetilde{V}^{(k-1)}(\widetilde{S}_{\tau+\delta_1})\right].\nn
\end{align}
From the general theory of optimal stopping we know that $e^{-qt}\widetilde{V}^{(k-1)}(\widetilde{S}_{t})$ is a supermartingale.
Thus, we have that
$$\E_s\left[e^{-q(\tau+\delta_1)}\widetilde{V}^{(k-1)}(\widetilde{S}_{\tau+\delta_1})\right]\leq\E_s\left[e^{-q(\delta_1)}\widetilde{V}^{(k-1)}(\widetilde{S}_{\delta_1})\right].$$
As $\widetilde{V}^{(1)}(s)=g(s)$ for $s\in[\mathbf{l}^*_1,\mathbf{u}^*_1]$, we get on $s\in[\mathbf{l}^*_1,\mathbf{u}^*_1]$
\begin{align}
\widetilde{V}^{(k)}(s)\leq g(s)+\E_s\left[e^{-q\delta_1}\widetilde{V}^{(k-1)}(\widetilde{S}_{\delta_1})\right]=g^{(k)}(s).\nn
\end{align}
This completes the proof.
\end{proof}

Let us define a set of stopping times for the recursive single stopping problems \eqref{Vk}-\eqref{recursion} as follows
\begin{align}
&\tau_1^*:=\inf\{t\geq 0:\ \widetilde{V}^{(1)}(\widetilde{S}_t)=g^{(1)}(\widetilde{S}_t)\},\label{tau*1}\\
&\tau_k^*:=\inf\{t\geq\tau_{k-1}^*+\delta_{k-1}:\ \widetilde{V}^{(k)}(\widetilde{S}_t)=g^{(k)}(\widetilde{S}_t)\}.\label{tau*k}
\end{align}
We prove below the equivalence between   \eqref{Vk} and the Swing problem \eqref{swing}. Thus, by the general optimal stopping theory,
we will prove that the set of stopping times
\begin{align}
\underline{\tau}^*:=(\tau_1^*,...,\tau_N^*)\label{tauset}
\end{align}
solves the Swing optimal stopping problem \eqref{swing}.
\begin{Thm}\label{optimalswing}
The Swing optimal stopping problem $V^{(N)}$ with $N$ exercise opportunities defined in \eqref{swing} with the put payoff function \eqref{payoff}
is equal to the recursive single stopping problems \eqref{Vk}-\eqref{recursion}. The optimal stopping rule (if it is finite)
for the Swing problem is given by $\underline{\tau}^*$ defined in \eqref{tauset}.
\end{Thm}
\begin{proof}
Firstly, note that $V^{(N)}(s)\leq\widetilde{V}^{(N)}(s)$, Indeed, by repeating below calculation $N$-times we have
\begin{align}
\widetilde{V}^{(N)}(s)=&\sup\limits_{\tau\in\mathcal{T}}\E_s\left[e^{-q\tau}g^{(N)}(\widetilde{S}_\tau)\right]\nn\\
=&\sup\limits_{\tau\in\mathcal{T}}\left(\E_s\left[e^{-q\tau}g(\widetilde{S}_\tau)\right]+\E_s\left[e^{-q\tau}\E_{\widetilde{S}_{\tau}}\left[e^{-q\delta_1}\widetilde{V}^{(N-1)}(\widetilde{S}_{\delta_1})\right]\right]\right)\nn\\
=&\sup\limits_{\tau\in\mathcal{T}}\left(\E_s\left[e^{-q\tau}g(\widetilde{S}_\tau)\right]+\E_s\left[e^{-q(\tau+\delta_1)}\sup\limits_{\nu\in\mathcal{T}}\E_{\widetilde{S}_{\tau+\delta_1}}\left[e^{-q\nu}g^{(N-1)}(\widetilde{S}_{\nu})\right]\right]\right)\nn\\
\geq&\sup\limits_{(\tau,\nu)\in\mathcal{T}^{(2)}}\left(\E_s\left[e^{-q\tau}g(\widetilde{S}_\tau)\right]+\E_s\left[e^{-q\nu}g^{(N-1)}(\widetilde{S}_\nu)\right]\right)\nn\\
\geq& \ldots \geq V^{(N)}(s).\nn
\end{align}
On the other hand, we
know that the set of stopping times $\{\tau_k^*\}_{k=1}^N$ defined in \eqref{tau*k} is optimal for the single recursion problem \eqref{Vk}. Therefore, we
can write
\begin{align}
\widetilde{V}^{(N)}(s)=&\E_s\left[e^{-q\tau_N^*}g^{(N)}(\widetilde{S}_{\tau_N^*})\right]
=\E_s\left[e^{-q\tau^*_N}g(\widetilde{S}_{\tau^*_N})\right]+\E_s\left[e^{-q\tau^*_N}\E_{\widetilde{S}_{\tau^*_N}}\left[e^{-q\delta_1}\widetilde{V}^{(N-1)}(\widetilde{S}_{\delta_1})\right]\right]\nn\\
\leq&\E_s\left[e^{-q\tau^*_N}g(\widetilde{S}_{\tau^*_N})\right]+\widetilde{V}^{(N-1)}(s)
=\E_s\left[\sum\limits_{i=1}^N e^{-q\tau^*_i}g(\widetilde{S}_{\tau^*_i})\right]\nn
\end{align}
where the above inequality follows from the supermartingale property of $e^{-qt}\widetilde{V}^{(N-1)}(\widetilde{S}_t)$. Summarizing, we get
\begin{align}
V^{(N)}(s)\leq\widetilde{V}^{(N)}(s)\leq\E_s\left[\sum\limits_{i=1}^N e^{-q\tau^*_i}g(\widetilde{S}_{\tau^*_i})\right].\nn
\end{align}
From the definition of $V^{(N)}$ as a supremum over all stopping times separated by the sequence $\{\delta_i\}$ it follows that
\[\E_s\left[\sum\limits_{i=1}^N e^{-q\tau^*_i}g(\widetilde{S}_{\tau^*_i})\right]\leq V^{(N)}(s).\]
This completes the proof.
\end{proof}
We now give a characterization of the set of optimal stopping rules for the Swing put option.
\begin{Thm}\label{optimalswing2}
The $k$-th optimal stopping rule $\tau_k^*$ is a first moment of entering some interval $I_k:=[\mathbf{l}^*_k,\mathbf{u}^*_k]$, i.e.
\begin{align}
\tau_k^*:=\{t\geq\tau_{k-1}^*+\delta_{k-1}:\ \widetilde{S}_t\in I_k\},\label{tau*k_lu}
\end{align}
where $\{\mathbf{l}^*_k\}_{k=1}^N$ and $\{\mathbf{u}^*_k\}_{k=1}^N$ satisfies the \textit{smooth fit} property.
Moreover, the sequence $\{I_k\}_{k=1}^N$ is nested, i.e. $[\mathbf{l}^*_k,\mathbf{u}^*_k]\subseteq [\mathbf{l}^*_{k+1},\mathbf{u}^*_{k+1}]$.
\end{Thm}
\begin{proof}
From Lemma \ref{lem:Vk_convex} we can conclude that for any $k\leq N$ the functions $\widetilde{V}^{(k)}$ and $g^{(k)}$ are non-increasing and convex.
Hence the optimal stopping times are of the form \eqref{tau*k_lu} for some sets $I_k$ being possibly sum of intervals (or half-lines).
To prove that $I_k$ are true intervals it is sufficient to prove that exercise regions are connected.
To do so assume to the contrary that the $k$th exercise region is not connected. That is,
we assume that there exists $s_1<s_2$ such that
\begin{equation}\label{ac}
\widetilde{V}^{(k)}(s_i)=g^{(k)}(s_i)\qquad\text{and}\qquad \widetilde{V}^{(k)}(s)>g^{(k)}(s)\quad\text{for $s\in(s_1,s_2)$.}\end{equation}
Let $s_0\in (s_1,s_2)$.
Now from the classical Wald-Bellman equation (see e.g. \cite[(2.1.6)-(2.1.7), p. 28]{Peskir}) or straightforward from the $k$th optimization problem \eqref{Vk}
\begin{equation}\label{HJB}
\widetilde{V}^{(k)}(s_0)\geq  \max \left\{g^{(k)}(s_0), \E_{s_0}\left[e^{-q\tau}g^{(k)}(\widetilde{S}_\tau)
\right]\right\}
\end{equation}
that holds true for any stopping time $\tau$.
Denote by $D_1$ the stopping region that is on the left from $s_1$ (jointly with $s_1$).
We choose $\tau=\inf\{t\geq 0: \widetilde{S}_t \in D_1\}$.
Now, as $g^{(k)}$ is non-increasing, observe that
\begin{equation}\label{keyidea}\widetilde{V}^{(k)}(s_0)\geq \E_{s_0}\left[e^{-q\tau}g^{(k)}(\widetilde{S}_\tau)\right]\geq g^{(k)}(s_1)\E_{s_0}\left[e^{-q\tau}\right]>g^{(k)}(s_1) \geq g^{(k)}(s_2). \end{equation}
This produces a contradiction with the fact
$\widetilde{V}^{(k)}$ is continuous and hence $\lim_{s_0\uparrow s_2}\widetilde{V}^{(k)}(s_0)=g^{(k)}(s_2)$.

In order to show that $\{I_k\}_{k=1}^N$ is nested we use inductive arguments. Let us define an auxiliary process
\begin{align*}
\overline{V}^{(k)}(\widetilde{S}_t):=e^{-qt}\left(\widetilde{V}^{(k)}(\widetilde{S}_t)-\widetilde{V}^{(k-1)}(\widetilde{S}_t)\right)\label{Vhat}
\end{align*}
for $k\in\{1,...,N\}$ with $\widetilde{V}^{(0)}(\widetilde{S}_t)=0$. For $k=1$ we have $\overline{V}^{(1)}(\widetilde{S}_0)=V(\widetilde{S}_0)$ and the function $\overline{V}^{(1)}(s)$ for $\widetilde{S}_0=s$ is equal to the value of the single American put option \eqref{Aput}. Therefore, for $k=1$, by the general theory of stopping problems, the process $\overline{V}^{(1)}(\widetilde{S}_t)$ is a supermartingale.
Moreover, from Lemma \ref{lem:Vk_lu} we know that $[\mathbf{l}^*_1,\mathbf{u}^*_1]\subseteq [\mathbf{l}^*_2,\mathbf{u}^*_2]$. Now, assume that for some $n$ the sequence $\{I_k\}_{k=1}^n$ is nested and the supermartingale property of $\overline{V}^{(h)}$ is satisfied for all $h\in\{1,...,n-1\}$, i.e.
\begin{align}
\left[\mathbf{l}^*_1,\mathbf{u}^*_1\right]\subseteq\left[\mathbf{l}^*_2,\mathbf{u}^*_2\right]\subseteq ... \subseteq\left[\mathbf{l}^*_n,\mathbf{u}^*_n\right]\nn
\end{align}
and
\begin{align}
\E_s\left[\overline{V}^{(h)}(\widetilde{S}_t)\right]\leq\overline{V}^{(h)}(s)=\widetilde{V}^{(h-1)}(s)-\widetilde{V}^{(h-2)}(s).\nn
\end{align}
Firstly, by induction we show that $\overline{V}^{(n)}(\widetilde{S}_t)$ is also a supermartingale. We define
\begin{align}
\omega_k^+:=\inf\{t\geq 0: \widetilde{S}_t\in [\mathbf{l}_k^*,\mathbf{u}_k^*]\},\qquad
\omega_k^-:=\inf\{t\geq 0: \widetilde{S}_t\notin [\mathbf{l}^*_k,\mathbf{u}^*_k]\}.\nn
\end{align}
From the general theory of optimal stopping we know that the stopped process $\overline{V}^{(n)}(\widetilde{S}_{t\wedge\omega_n^+})$ is a martingale because $e^{-qt\wedge\omega_n^+}\widetilde{V}^{(h)}(\widetilde{S}_{t\wedge\omega_n^+})$ are martingales for all $h$. Furthermore, the process
$$\overline{V}^{(n)}(\widetilde{S}_{t\wedge\omega_n^-\wedge\omega_{n-1}^+})=e^{-qt\wedge\omega_n^-\wedge\omega_{n-1}^+}\widetilde{V}^{(n)}(\widetilde{S}_{t\wedge\omega_n^-\wedge\omega_{n-1}^+})-e^{-qt\wedge\omega_n^-\wedge\omega_{n-1}^+}\widetilde{V}^{(n-1)}(\widetilde{S}_{t\wedge\omega_n^-\wedge\omega_{n-1}^+})$$
is a stopped supermartingale less a stopped martingale, hence a supermartingale. Finally, $\overline{V}^{(n)}(\widetilde{S}_{t\wedge\omega_{n-1}^-})$ is a supermartingale from the definition of $\widetilde{V}^{(k)}$ and independence of $\delta_k$ and $S$
\begin{align}
\E_s\left[\overline{V}^{(n)}(\widetilde{S}_{t\wedge\omega_{n-1}^-})\right]=&\E_s\left[\widetilde{V}^{(n)}(\widetilde{S}_{t\wedge\omega_{n-1}^-})-\widetilde{V}^{(n-1)}(\widetilde{S}_{t\wedge\omega_{n-1}^-})\right]=\E_s\left[g^{(n)}(\widetilde{S}_{t\wedge\omega_{n-1}^-})-g^{(n-1)}(\widetilde{S}_{t\wedge\omega_{n-1}^-})\right]\nn\\
=& \E_s\left[e^{-q(\delta_1+t\wedge\omega_{n-1}^-)}\left(\widetilde{V}^{(n-1)}(\widetilde{S}_{\delta_1+t\wedge\omega_{n-1}^-})-\widetilde{V}^{(n-2)}(\widetilde{S}_{\delta_1+t\wedge\omega_{n-1}^-})\right)\right]=\E_s\left[\overline{V}^{(n-1)}(\widetilde{S}_{\delta_1+t\wedge\omega_{n-1}^-})\right]\nn\\
\leq & \E_s\left[\overline{V}^{(n-1)}(\widetilde{S}_{\delta_1})\right]=\E_s\left[e^{-q\delta_1}\left(\widetilde{V}^{(n-1)}(\widetilde{S}_{\delta_1})-\widetilde{V}^{(n-2)}(\widetilde{S}_{\delta_1})\right)\right]=g^{(n)}(s)-g^{(n-1)}(s)\nn\\
=&\widetilde{V}^{(n)}(s)-\widetilde{V}^{(n-1)}(s)=\overline{V}^{(n)}(s).\nn
\end{align}
All these cases together with the \textit{smooth fit} property at all $\mathbf{l}^*_k$ and $\mathbf{u}^*_k$ give us the supermartingale property of the process $\overline{V}^{(n)}(\widetilde{S}_t)$. The \textit{smooth fit} is satisfied by the same arguments as the ones presented in Proposition \ref{smooth_fit} for the single American put option.

Now, using proved supermartingale property of $\overline{V}$ we show that $I_{n}\subseteq I_{n+1}$. Indeed, for $s\in[\mathbf{l}^*_n,\mathbf{u}^*_n]$ we have:
\begin{align}
\widetilde{V}^{(n+1)}(s)=&\sup\limits_{\tau\in\mathcal{T}}\E_s\left[e^{-q\tau}g^{(n+1)}(\widetilde{S}_\tau)\right]\leq\sup\limits_{\tau\in\mathcal{T}}\E_s\left[e^{-q\tau}g^{(n)}(\widetilde{S}_\tau)\right]+\sup\limits_{\tau\in\mathcal{T}}\E_s\left[e^{-q\tau}\left(g^{(n+1)}(\widetilde{S}_\tau)-g^{(n)}(\widetilde{S}_\tau)\right)\right]\nn\\
=&\widetilde{V}^{(n)}(s)+\sup\limits_{\tau\in\mathcal{T}}\E_s\left[e^{-q(\tau+\delta_1)}\left(\widetilde{V}^{n}(\widetilde{S}_{\tau+\delta_1})-\widetilde{V}^{(n-1)}(\widetilde{S}_{\tau+\delta_1})\right)\right]=\widetilde{V}^{(n)}(s)+\sup\limits_{\tau\in\mathcal{T}}\E_s\left[\overline{V}^{(n)}(\widetilde{S}_{\tau+\delta_1})\right]\nn\\
\leq & \widetilde{V}^{(n)}(s)+\E_s\left[\overline{V}^{(n)}(\widetilde{S}_{\delta_1})\right]=\widetilde{V}^{(n)}(s)+(g^{(n+1)}(s)-g(s))-(g^{(n)}(s)-g(s))\nn\\
=&\widetilde{V}^{(n)}(s)-g^{(n)}(s)+g^{(n+1)}(s)=g^{(n+1)}(s).\nn
\end{align}
This proves that $\widetilde{V}^{(n+1)}(s)=g^{(n+1)}(s)$ and hence $[\mathbf{l}_{n}^*,\mathbf{u}^*_{n}]\subseteq [\mathbf{l}^*_{n+1},\mathbf{u}^*_{n+1}]$.
\end{proof}

\bigskip
\section{The put-call symmetry}\label{sec:putcall}
Let $\widetilde{X}_t$  be a completely asymmetric L\'evy process. We recall that
 $\widetilde{X}_t$ has the Laplace exponent \eqref{eq:exponent}
 $$\widetilde{\psi}(\phi)=\mu\phi
+\frac{\sigma^2}{2} \phi^{2}+\int_{\mathbb{R}\backslash\{0\}}\big(\mathrm e^{\phi
u}-1-\phi u\mathbbm{1}_{(|u|<1)}\big)\Pi(\diff u),$$
where $\mu=q-\delta -\frac{\sigma^2}{2} - \int_{\mathbb{R}}\big(\mathrm e^{
u}-1- u\mathbbm{1}_{(|u|<1)}\big)\Pi(\diff u)$
by \eqref{esschertransform}. We recall that by $\delta$ we denote a dividend rate.
The process $\widetilde X$ can be identified via its triplet  $(\mu, \sigma, \Pi)$.

In this section we will find a relationship between the values and the exercise regions of  perpetual American call and put options, the so-called
{\it put-call symmetry}.
We denote by
\begin{equation*}
V_p(s, K, q, \delta, \sigma, \Pi) :={\sup\limits_{\tau\in\mathcal{T}} }\E_s\left[e^{-q\tau}\left(K-\widetilde{S}_{\tau}\right)^+ \right] = {\sup\limits_{\tau\in\mathcal{T}}} \E_{(\log s)}\left[e^{-q\tau}\left(K-e^{\widetilde{X}_{\tau}}\right)^+ \right]
 \end{equation*}
the value of the perpetual American put option
and by
\[V_c(s, K, q, \delta, \sigma, \Pi) :={\sup\limits_{\tau\in\mathcal{T}} }\E_s\left[e^{-q\tau}\left(\widetilde{S}_{\tau}-K\right)^+ \right] = {\sup\limits_{\tau\in\mathcal{T}}} \E_{(\log s)}\left[e^{-q\tau}\left(e^{\widetilde{X}_{\tau}}-K\right)^+ \right]
\]
the value of the perpetual American call option.

If the discounting rate $q<0$, the arguments given in Section \ref{sec:valuation} (see also \cite{Battauz}),
show that the perpetual American put option either is never exercised or admits a double continuation region identified by the two critical prices
\begin{equation*}\label{lp}
\mathbf{l}_p^*=\mathbf{l}^*=e^{l_p^*}=e^{l^*} \leq  e^{u^*}=e^{u_p^*} =\mathbf{u}^*=\mathbf{u}_p^*\leq K.
\end{equation*}

%
Let
\begin{align}
v_p(s, K, q, \delta, \sigma, \Pi, l,u)&:=\E_{(\log s)}\left[e^{-q\tau_{l,u}}\left(K-e^{\widetilde{X}_{\tau_{l,u}}}\right)^+ \right]\nn \\
v_c(s, K, q, \delta, \sigma, \Pi, l,u)&:=\E_{(\log s)}\left[e^{-q\tau_{l,u}}\left(e^{\widetilde{X}_{\tau_{l,u}}}-K\right)^+ \right]\nn
\end{align}
for $\tau_{l,u}$ given in \eqref{tau} and define $\widecheck{\Pi}(\diff y):= e^{-y} \Pi(-\diff y).$

\begin{Lem}\label{symmetry_perpetual}  
For all $\log K< l< u $   we have
\begin{align} \label{intervalsimmetry}
v_c(s, K, q, \delta, \sigma, \Pi, l,u) = v_p(K, s, \delta, q, \sigma, \widecheck{\Pi}, \log s + \log K - u, \log s + \log K - l ).
\end{align}
\end{Lem}
\begin{proof} We prove the result for $\widetilde X=X$ being spectrally negative. The proof for the spectrally positive case is similar.
When $X$ is spectrally negative, similar arguments to those employed in Lemma  \ref{v_value}  show that for $\log K < l$ we have
\begin{align*}
v_c(e^x, K, q, \delta, \sigma, \Pi, l,u) = &\left(e^x-K\right)\mathbbm{1}_{(x\in [l,u])}+\left(e^l-K \right)e^{-\Phi(q)(l-x)}\mathbbm{1}_{(x<l)}\label{eq:call_v_value}\\
&+\Bigg\{\int_0^\i\int_{(0,\i)} e^{-\Phi(q)(l-u+y)\vee 0}\left(e^{l\vee (u-y)}-K\right)r^{(q)}(x-u,z)\Pi(-z-\diff y)\diff z\nn\\
&\qquad+\left(e^u-K\right)\frac{\sigma^2}{2}\left(\Wp(x-u)-\Phi(q)W^{(q)}(x-u)\right)\Bigg\}\mathbbm{1}_{(x>u)}\nn.
\end{align*}
Denote
\begin{equation*}\label{checkX}
\widecheck{X}_t= -(X_t - \log s) +\log K\end{equation*}
and $\widecheck x = \log K, \widecheck K = e^x, \widecheck q = \delta, \widecheck l =x + \log K - u, \widecheck u = x + \log K - l $.
Note that $\widecheck{X}_t$ is the spectrally positive process starting from $\widecheck x$ and defined by the triplet $(\widecheck\mu, \sigma, \widecheck{\Pi})$, with $\widecheck\mu = -\mu-\sigma^2$.
As shown in the proof of Lemma 2.1 in \cite{FajardoMordecki}, its characteristic exponent $\widecheck \psi$ is related to $\psi$ by the following identity $\widecheck \psi(\phi) = \psi(1-\phi) -\psi(1)$.
The process $\mathring{X}_t:=\widehat{\widecheck{X}}_t=(-\widecheck X)$ is spectrally negative with the Laplace exponent  $\mathring{\psi}(\phi)=\widecheck \psi (-\phi)$, which
right continuous inverse function $\mathring{\Phi}$ is given by
\begin{align}
\mathring{\Phi}(\widecheck q) = \Phi(q)-1.\nn
\end{align}
\noindent Indeed,
note that $\mathring{\psi}(\mathring{\Phi}(\widecheck q))=\widecheck \psi(-\mathring{\Phi}(\widecheck q))= \widecheck \psi (- ( \Phi(q)-1)) = \widecheck \psi (1-\Phi(q)) = \psi(\Phi(q)) - \psi(1) = q - (q-\delta) = \delta = \widecheck q.$

Finally,  Lemma 8.4 of 
\cite{KIntr} gives us that the first scale function $\mathring W^{(q)} (x)$ of $\mathring{X}_t$ satisfies
\begin{align}
W^{(q)}(x) =e^x \mathring W^{(q - \psi(1))} (x) = e^{x} \mathring W^{(\mathring q )} (x) \qquad \quad \Wp(x) =  e^{x} \mathring W^{(\mathring q )} (x) + e^x W'^{(\mathring q )} (x)\nn
\end{align}
and its resolvent density $\mathring r^{(\mathring q)}  (x, y)$ satisfies
\begin{align}
r^{(q)}(x,y) = e^{(-\mathring \Phi (\mathring q) - 1)y}  \,  e^{x} \,  \mathring W^{(\mathring q )} (x) -  e^{x - y} \,  \mathring W^{(\mathring q )} ( x - y) =  e^{x - y} \,  \mathring r^{(\mathring q)}  (x, y).\nn
\end{align}
Using the above relations and Lemma \ref{v_value} we can straightforward derive equation \eqref{intervalsimmetry}.
Thus the assertion of this lemma holds true.
\end{proof}

Lemma \ref{symmetry_perpetual} immediately produces the following {\it put-call symmetry for perpetual American options}.
\begin{Thm}\label{dualsym} 
Assume that the stopping region of the perpetual put option is not empty, with optimal stopping boundaries
$\mathbf{l}_p^*=e^{l_p^*}$ and $ \mathbf{u}_p^*=e^{u_p^*}$.
Then,
\begin{equation*}
 \label{perpetualsimmetry}
V_c(s, K, q, \delta, \sigma, \Pi) = V_p(K, s, \delta, q, \sigma, \widecheck{\Pi}).
\end{equation*}
Moreover, the American call option admits a double continuation
region with optimal stopping boundaries $\mathbf{l}_c^*=e^{l_c^*}, \mathbf{u}_c^*=e^{u_c^*}$, such that
$$\mathbf{l}_c^*\mathbf{u}_p^*  = \mathbf{l}_p^* \mathbf{u}_c^*  = sK$$
or equivalently
$$ l^*_c  + u^*_p  = l^*_p  + u^*_c  = \log s + \log K.$$
\end{Thm}

\begin{Rem}\rm
A similar result for the Black-Scholes model is derived in 
\cite{DeDonno}.
For the standard case $q> 0$, Theorem \ref{dualsym}   is proved by
\cite{CarrChesney} in a diffusive market and by 
\cite{Schroder}
 in a general semimartingale model.   A precise description of the duality in the L\'evy market for $q>0$ can be found in 
 \cite[Cor. 2.2]{FajardoMordecki}.
A comprehensive review of the put-call duality for American options is given in 
\cite{Detemple}.
\end{Rem}

\begin{proof}
Using similar arguments like in the proof of Lemma \ref{newlemma} one can conclude that
for the call American option the stopping region has to be also an interval or a half-line.
Now, Lemma \ref{symmetry_perpetual} implies that
\begin{eqnarray*}
\frac{\partial}{\partial l}v_c(s, K, q, \delta, \sigma, \Pi, l,u) = - \frac{\partial}{\partial u}
v_p(K, s, \delta, q, \sigma, \widecheck{\Pi}, \log s + \log K - u, \log s + \log K - l ), \\
\frac{\partial}{\partial u}v_c(s, K, q, \delta, \sigma, \Pi, l,u)  = - \frac{\partial}{\partial l}  v_p(K, s, \delta, q, \sigma, \widecheck{\Pi}, \log s + \log K - u, \log s + \log K - l).
\end{eqnarray*}
All above derivatives are equal zero (hence maximizing appropriate option values)
if and only if  $ l^*_p =  -u^*_c +\log s + \log K
$, $   u^*_p  = -  l^*_c + \log s + \log K.
$
This completes the proof.
\end{proof}

A put-call symmetry result can be derived also  for Swing options as a consequence of the previous theorem.
Let
\begin{align}
V^{(N)}_p(s, K, q, \delta, \sigma, \Pi, \underline{\tau}^*_p)&:=\sup\limits_{\underline{\tau}\in\mathcal{T}^{(N)}}
\E_{(\log s)}\left[\sum\limits_{i=1}^N e^{-q\tau_i}\left(K-e^{\widetilde{X}_{\tau_i}}\right)^+\right]=\E_{(\log s)}\left[\sum\limits_{i=1}^N e^{-q\tau^*_{i,p}}\left(K-e^{\widetilde{X}_{\tau^*_{i,p}}}\right)^+\right],\nn\\
V^{(N)}_c(s, K, q, \delta, \sigma, \Pi, \underline{\tau}^*_c)&:=\sup\limits_{\underline{\tau}\in\mathcal{T}^{(N)}}\E_{(\log s)}\left[\sum\limits_{i=1}^N e^{-q\tau_i}\left(e^{\widetilde{X}_{\tau_i}}-K\right)^+ \right]=\E_{(\log s)}\left[\sum\limits_{i=1}^N e^{-q\tau^*_{i,c}}\left(e^{\widetilde{X}_{\tau^*_{i,c}}}-K\right)^+ \right],\nn
\end{align}
be put and call American options, respectively, where $\mathcal{T}^{(N)}$ is defined in \eqref{Tn} and
$$\underline{\tau}^*=(\tau^*_{1},...,\tau^*_{N})=\underline{\tau}_p^*=(\tau^*_{1,p},...,\tau^*_{N, p}),\qquad \underline{\tau}_c^*=(\tau^*_{1,c},...,\tau^*_{N, c})$$
are the corresponding  optimal stopping rules.
From Theorems \ref{optimalswing} and \ref{optimalswing2} it follows that
\begin{align}
\tau_i^*=\tau^*_{i,p}=\inf\{t\geq 0:\ \widetilde{X}_t\in[l^*_i,u^*_i]\},\nn
\end{align}
where $\mathbf{l}^*_i=e^{l^*_i}=e^{l^*_{i,p}} \leq  e^{u_i^*} = e^{u_{i,p}^*}=\mathbf{u}^*_i$ for $\mathbf{l}^*_i$ and $\mathbf{u}^*_i$ defined in Theorem \ref{optimalswing2}.
\begin{Cor}
The following put-call symmetry holds
\begin{align*} \label{swing_symmetry}
V^{(N)}_c(s, K, q, \delta, \sigma, \Pi, \underline{\tau}^*_c) = V^{(N)}_p(K, s, \delta, q, \sigma, \widecheck{\Pi}, \underline{\tau}^*_p)
\end{align*}
and
\begin{equation}
\label{callentrance}
\tau^*_{i,c}:=\inf\{t\geq 0:\widetilde{X}_t\in[l^*_{i,c},u^*_{i,c}]\}
\end{equation}
with
\begin{align}
&l^*_{i,c}:=\log s + \log K - u_{i,p}^*,\nn\\
&u^*_{i,c}:=\log s + \log K - l_{i,p}^*.\nn
\end{align}
\end{Cor}
\begin{proof}
To prove that stopping regions for the call Swing option are of the form of \eqref{callentrance} one can use similar idea like the one given in  \eqref{keyidea} but
in this case one has to exchange roles of $s_1$ and $s_2$ as payoff functions $g^{(k)}$ are all non-decreasing.
Now, the result follows immediately from Lemma \ref{symmetry_perpetual} as symmetry holds for all $l,u$ such that $\log K<l<u$, given also that the Swing option is a sum of single options.
\end{proof}

\bigskip
\section{Numerical analysis}\label{sec:examples}
\subsection {Black-Scholes model}
Let $$X_t = x+\mu t + \sigma W_t$$ be a linear Brownian motion where $W_t$ is a Brownian motion, $\mu = q-\delta-\frac{\sigma^2}{2}$ is a drift under a martingale measure
(see \eqref{esschertransform}), $\sigma>0$ is a volatility and $x=\log s$ is the starting position of $X$. As the linear Brownian motion do not have any jumps, it is the only L\'evy process which is at the same time spectrally negative and spectrally positive.
This model corresponds the seminal Black-Scholes market model.

For the linear Brownian motion $X_t$ the Laplace exponent and the scale functions are given by
\begin{align}
\psi(\phi)=\mu\phi+\frac{1}{2}\sigma^2\phi^2\nn
\end{align}
and
\begin{align}
W^{(q)}(\phi)&=\frac{1}{\Xi\sigma^2}\left(e^{(\Xi-\frac{\mu}{\sigma^2})\phi}-e^{(-\Xi-\frac{\mu}{\sigma^2})\phi}\right)=\frac{2}{\Xi\sigma^2}e^{-\frac{\mu}{\sigma^2}\phi}\sinh(\Xi\phi),\nn\\
\end{align}
where
\begin{align}
\Xi=\frac{\sqrt{\mu^2+2q\sigma^2}}{\sigma^2}.\nn
\end{align}

Using Lemma \ref{v_value}, we now give the value of the American put option with the negative discounting factor $q<0$. Firstly, we calculate the right inverse Laplace transform
\begin{align}
\Phi(q)=-\frac{\mu}{\sigma^2}+\Xi.\nn
\end{align}
It is well-defined when $\mu^2+2q\sigma^2\geq 0$ and for $q<0$ it takes negative value. Thus, the American put option with the negative rate admits the double continuation region and its value for the Brownian motion reduces to:
\begin{align}
v(x,l^*,u^*)=&\left(K-e^x\right)\mathbbm{1}_{(x\in [l^*,u^*])}+\left(K-e^{l^*}\right)e^{(\frac{\mu}{\sigma^2}-\Xi)(l^*-x)}\mathbbm{1}_{(x<l^*)}\nn\\
&+\left(K-e^{u^*}\right)e^{-(\frac{\mu}{\sigma^2}+\Xi)(x-u^*)}\mathbbm{1}_{(x>u^*)},\nn
\end{align}
where the optimal levels $l^*$ and $u^*$
could be identified using either first-order condition or the smooth fit principle given in Proposition \ref{smooth_fit}, and
they are given by
\begin{align}
l^*=&\log\left( K\frac{\Phi(q)}{\Phi(q)-1}\right),\nn\\
u^*=&\log\left( K\frac{2\Xi-\Phi(q)}{2\Xi-\Phi(q)+1}\right).\nn
\end{align}
\begin{Rem}\rm
The optimal level $l^*$ solves
\begin{equation*}\label{lder} \frac{\partial}{\partial l} v(x,l^*,u^*)=-e^{l^*}  e^{-\Phi(q)(l^*-x)} -\Phi(q) (K-e^{l^*}) e^{-\Phi(q)(l^*-x)} = 0
\end{equation*}
because it maximizes the value function of the American put option.
The above equation is equivalent to
$$ \Phi(q) (K-e^{l^*}) = -e^{l^*}
$$
which does not have solution when $q\geq 0$. This means that for $q\geq 0$ we always have a single continuation region; see Table \ref{Table1}.
In fact, 
\cite{AliliKyprianou} (see also 
\cite[Th. 9.2]{KIntr}) proved that
this stopping region is $(-\infty,u^*]$, hence the continuation region $\mathcal{C}$ equals $(u^*,\infty)$ in this case.
\end{Rem}
\begin{figure}[!ht]
\centering
\includegraphics[width=0.6\textwidth]{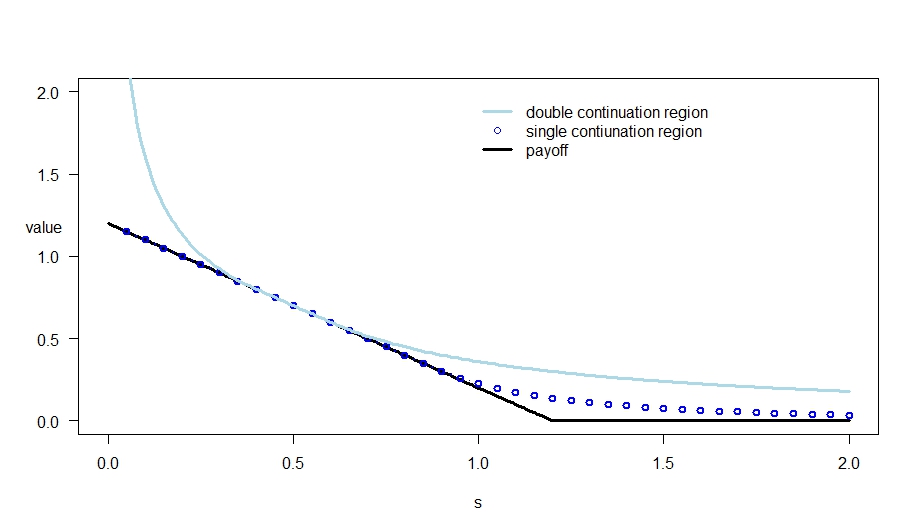}
\caption{\footnotesize{The American put option value for different discounting factor: $q=-0.01$ with double continuation region and $q=0.01$ with single continuation region. Parameters: $\sigma=0.2, \delta=-0.06$.}}
\label{fig:Aput}
\end{figure}

In Table \ref{Table1}
 we give a summary for all possible continuation regions with respect to parameters for the American put option in the Brownian motion case. The put-call symmetry yields an analogous result for the American call option, which we describe  in Table \ref{Table2}.
Note that in the case of the American call option one can derive only single stopping region when $q<0$ and $\delta >0$.
This case corresponds to the conditions considered by 
\cite{Xia&Zhou}.
\begin{table}[!htbp]
\centering
\caption{The continuation region for the American put option with respect to discounting rate.}
\label{Table1}
\begin{tabular}{|l|l|l|l|}
\toprule
\multicolumn{3}{|c|}{\bf{discounting rate}} & \multicolumn{1}{|c|}{\bf{continuation region $\mathcal{C}$}} \\
\midrule
\multicolumn{3}{|l|}{$q>0$} & single continuation region; $\mathcal{C}=(u^*,\i)$\\
\hline
\multirow{2}{*}{$q=0$} & \multicolumn{2}{|l|}{$\mu>0$} & single continuation region; $\mathcal{C}=(u^*,\infty)$\\
\cline{2-4}
& \multicolumn{2}{|l|}{$\mu\leq 0$} &no early exercise; $\mathcal{C}=\mathbb{R}$\\
\hline
\multirow{5}{*}{$q<0$} & \multirow{2}{*}{$\mu^2+2q\sigma^2>0$} & $\mu>0$ & double continuation region; $\mathcal{C}=(-\i,l^*)\cup (u^*,\i)$\\
\cline{3-4}
& & $\mu\leq 0$ & no early exercise; $\mathcal{C}=\mathbb{R}$\\
\cline{2-4}
& \multirow{2}{*}{$\mu^2+2q\sigma^2=0$} & $\mu>0$ &double continuation region; $l^*=u^*$; $\mathcal{C}=(-\i,l^*)\cup (l^*,\i)$\\
\cline{3-4}
& & $\mu\leq 0$ & no early exercise; $\mathcal{C}=\mathbb{R}$\\
\cline{2-4}
& \multicolumn{2}{|l|}{$\mu^2+2q\sigma^2<0$} & no early exercise; $\mathcal{C}=\mathbb{R}$\\
\bottomrule
\end{tabular}
\end{table}

\begin{table}[!htbp]
\centering
\caption{The continuation region for the American call option with respect to discounting rate.}
\label{Table2}
\begin{tabular}{|l|l|l|l|l|}
\toprule
\multicolumn{4}{|c|}{\bf{discounting rate}} & \multicolumn{1}{|c|}{\bf{continuation region $\mathcal{C}$}} \\
\midrule
\multirow{2}{*}{$q\geq 0$} & \multicolumn{3}{|l|}{$\delta>0$} & single continuation region; $\mathcal{C}=(-\i,l^*)$\\
\cline{2-5}
& \multicolumn{3}{|l|}{$\delta\leq 0$} &no early exercise; $\mathcal{C}=\mathbb{R}$\\
\hline
\multirow{7}{*}{$q<0$} & \multicolumn{3}{|l|}{$\delta>0$} & single continuation region; $\mathcal{C}=(-\i,l^*)$\\
\cline{2-5}
& \multicolumn{3}{|l|}{$\delta=0$} & no early exercise; $\mathcal{C}=\mathbb{R}$\\
\cline{2-5}
& \multirow{5}{*}{$\delta<0$} & \multirow{2}{*}{$\mu^2+2\delta\sigma^2>0$} & $\mu<0$ & double continuation region; $\mathcal{C}=(-\i,l^*)\cup (u^*,\i)$\\
\cline{4-5}
& & & $\mu\geq 0$ & no early exercise; $\mathcal{C}=\mathbb{R}$\\
\cline{3-5}
& & \multirow{2}{*}{$\mu^2+2\delta\sigma^2=0$} & $\mu<0$ & double continuation region; $l^*=u^*$; $\mathcal{C}=(-\i,l^*)\cup (l^*,\i)$\\
\cline{4-5}
& & & $\mu\geq 0$ & no early exercise; $\mathcal{C}=\mathbb{R}$\\
\cline{3-5}
& & \multicolumn{2}{|l|}{$\mu^2+2\delta\sigma^2<0$} & no early exercise; $\mathcal{C}=\mathbb{R}$\\
\bottomrule
\end{tabular}
\end{table}

By numerical methods, we also find the optimal stopping regions for the Swing option in the Black-Scholes model.
In Figure \ref{fig:l&u} we present the sequence of the optimal levels $l_k^*$ and $u^*_k$ calculated
by Monte Carlo methods with $10.000$ iterations in each recursive step.

\begin{figure}[!ht]
\centering
\includegraphics[width=0.45\textwidth]{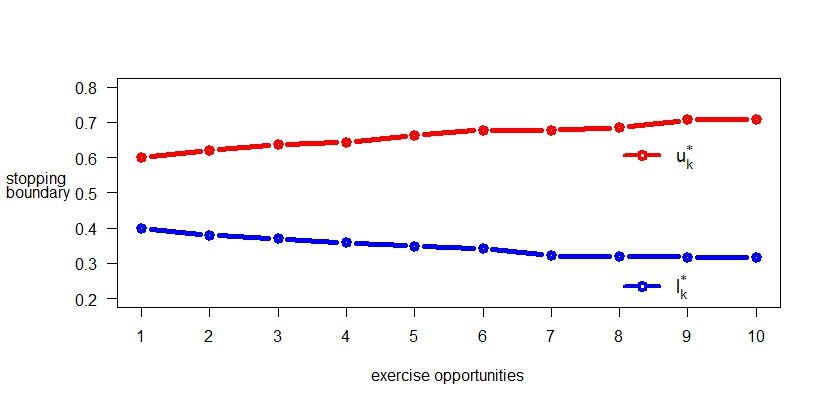}
\caption{\footnotesize{The optimal levels $l_k^*$ and $u_k^*$ for $k=1,...,10$ for a Swing American put option with negative discounting rate. Parameters: $q=-0.01,\ K=1.2,\ \sigma=0.2, \delta=-0.06$, $\delta_k=0.5$.}}
\label{fig:l&u}
\end{figure}

\subsection{Jump-diffusion model}
For comparison purpose we extend the numerical calculation of the Black-Scholes model given above by adding the possibility of downward jumps for the log prices.
We model this jumps by an exponential distribution with parameter $\rho$ and we assume that they occur in the market with fixed intensity $\lambda>0$.
Formally, we model the logarithm of the risky underlying asset price $X_t=\log S_t$ by
\begin{align}
X_t=x+\mu t+\sigma W_t-\sum\limits_{i=1}^{N_t}\xi_i,\label{Compound}
\end{align}
where,  $x=\log s$ and $\mu=q-\delta -\frac{\sigma^2}{2} +\frac{\lambda\rho}{\phi+\rho}-\lambda$,
$N_t$ is a Poisson process with intensity $\lambda>0$ and $\{\xi_i\}_{\{i=1\}}^\infty$ is an sequence of i.i.d. exponential jumps with parameter $\rho$.
It is well-known that jump-diffusion models better reproduces a wide variety of implied volatility skews and smiles, and market fluctuation (see the discussion and references in the introduction).

Note that the process $X_t$ given in \eqref{Compound} is a spectrally negative L\'evy process with
 Laplace exponent
\begin{align*}
\psi(\phi)=\mu\phi+\frac{\sigma^2}{2}\phi^2+\frac{\lambda\rho}{\phi+\rho}-\lambda.\label{comppsi}
\end{align*}
In order for the American put option stopping problem to be well-defined, following Lemma \ref{lem:condition}  we assume that
\begin{equation*}\label{stability}\E_{(x)}X_1= \psi'(0) = \mu-\frac{\lambda}{\rho}>0.
\end{equation*}

\noindent  The scale function is given by:
\begin{align}
W^{(q)}(x)=\frac{e^{\Phi(q)x}}{\psi^\prime(\Phi(q))}+\frac{e^{-\varphi_1x}}{\psi^\prime(-\varphi_1)}+\frac{e^{-\varphi_2x}}{\psi^\prime(-\varphi_2)},\nn
\end{align}
where (possibly complex) valued $\varphi_i$ for $i=1,2$ satisfy
\begin{align}
\psi(-\varphi_i)=q.\label{three}
\end{align}
One can observe that the above equation is equivalent to a cubic equation and it can be solved using Cardano roots.
This type of equation either have all three roots real or
one real root and two non-real complex conjugate roots.
The problem of the American put option with the negative rate is well-defined when $\Phi(q)$ exists.
This is equivalent to having three real solutions of \eqref{three}.
The roots of $\psi(\phi)=q$ for the negative discounting rate are presented in the Figure \ref{fig:psi}.
Note that $\Phi(q)$ has also  to satisfy the following inequality
\begin{align}
\Phi(q)+\rho>0.\label{addin}
\end{align}
Indeed,  as $\psi(\Phi(q))=q$, we can write
\begin{align}
e^q=e^{\psi(\Phi(q))X_1}=\E\Big[e^{\Phi(q)}\Big]=\E\Big[e^{\Phi(q)(\mu t+\sigma W_t)}\Big]\E\Big[e^{-\Phi(q)\sum\limits_{i=1}^{N_t}\xi_i}\Big]
=e^{\mu\Phi(q)+\frac{1}{2}\sigma^2\Phi^2(q)}e^{-\lambda(1-\int_0^\i e^{-\Phi(q)\xi}\rho e^{-\rho\xi}\diff\xi)}.\nn
\end{align}
Now, the integral in the second expectation is finite if and only if \eqref{addin} holds true.

\begin{figure}[!ht]
\center{
\subfloat{\includegraphics[width=0.45\textwidth]{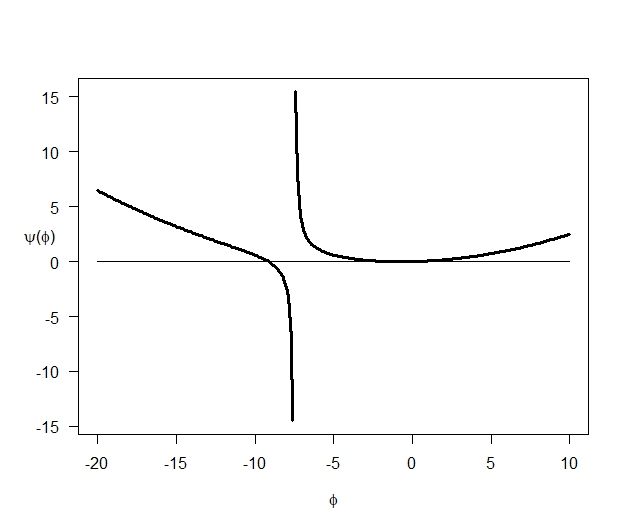}}
\quad
\subfloat{\includegraphics[width=0.45\textwidth]{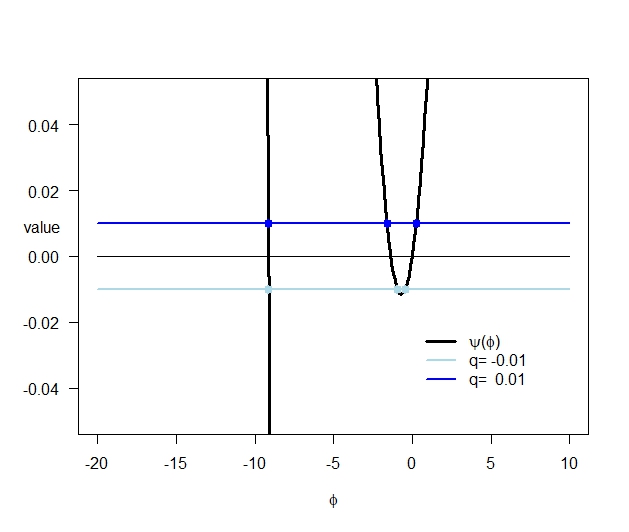}}
\caption{\footnotesize{The Laplace transform $\psi(\phi)$ (left) and solutions of $\psi(\phi)=q$ for different $q$ (right). Parameters: $\sigma=0.2, \lambda=0.2,\ \rho=7.5,\ \mu=0.06$.}}\label{fig:psi}}
\end{figure}

We can find  the value of the American put option using   \eqref{eq:v_value}. To do so, it remains only to
  to calculate the double integral appearing in the formula. As the jumps have exponential distribution, it  can be computed as a product of two integrals as follows:
\begin{align}
\int_0^\i e^{-\Phi(q)(l-u+y)\vee 0}\left(K-e^{l\vee (u-y)}e^{-\rho y}\diff y\right)=&\left(K-e^l\right)\frac{e^{-\rho (u-l)}}{\Phi(q)+\rho}+\frac{1}{\rho}K\left(1-e^{-\rho (u-l)}\right)\nn\\
&-\frac{1}{1+\rho}e^u\left(1-e^{-(1+\rho)(u-l)}\right)\nn
\end{align}
and
\begin{align}
&\int_0^\i \left(e^{-\Phi(q)z}W^{(q)}(x-u)-W^{(q)}(x-u-z)\right)\lambda\rho e^{-\rho z}\diff z =\frac{\lambda\rho}{\Phi(q)+\rho}W^{(q)}(x-u)\nn\\
&\qquad\qquad+\frac{\lambda\rho}{\Phi(q)+\rho}\frac{e^{-\rho(x-u)}-e^{\Phi(q)(x-u)}}{\psi^\prime (\Phi(q))}+\frac{\lambda\rho}{\rho-\varphi_1}\frac{e^{-\rho(x-u)}-e^{-\varphi_1(x-u)}}{\psi^\prime (-\varphi_1)}+\frac{\lambda\rho}{\rho-\varphi_2}\frac{e^{-\rho(x-u)}-e^{-\varphi_2(x-u)}}{\psi^\prime (-\varphi_2)}.\nn
\end{align}
The optimal boundaries can then be derived by using the smooth fit principle or first-order conditions.
The value functions of the American put option and the optimal stopping regions for some negative discounting factor (double continuation region) and for some positive discounting factor (single continuation region) are presented in Figure \ref{fig:ex_jumps}.
\begin{figure}[!ht]
\centering
\includegraphics[width=0.6\textwidth]{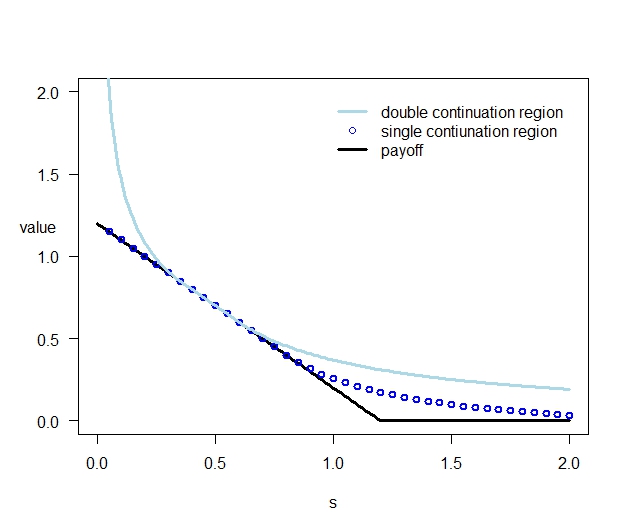}
\caption{\footnotesize{The American put option value for different discounting factor: $q=-0.01$ with double continuation region and $q=0.01$ with single continuation region. Parameters: $\sigma=0.2, \lambda=0.2,\ \rho=7.5,\ \mu=0.06$.}}
\label{fig:ex_jumps}
\end{figure}

To understand the impact of jumps on the value function and on the stopping region we compare various intensities of exponential distribution $\rho$ and various intensities of arrival rate $\lambda$, leaving fixed the other parameters, in Figure \ref{fig:ex_jumps2} and Figure \ref{fig:ex_jumps_lambda}, respectively.
Note that the increase in $\rho$ corresponds to a decrease in the average sizes of the jumps, which, we recall, are downward jumps.
On the other hand, the increase in $\lambda$ impacts in increase of average number of observed jumps in one time unit.

Note that the value function gets lower and the stopping region gets larger as the parameter of the exponential distribution increases. The same behaviour is observed when the arrival intensity decreases. As for  an American put option the jump reduces the asset price a higher or more frequent downward jumps increase the likelihood of a higher gain, hence the price of the option. As for the critical prices, we recall that when the interest rate is negative, it is not convenient to  exercise  the option  if it  is either too deep or not enough in the money.  If the option is in the money, a higher or more frequent downward jumps could drastically reduce  the asset price, making the exercise no longer convenient, therefore a safer (namely higher) lower level for the stopping region is required. The behaviour of the upper critical value, which for the put option is the "standard" one, agrees with the literature for the case $q\ge 0$. For instance 
\cite{Amin} shows that
jumps may reduce the value of the asset, early exercise is postponed and the option is exercised if the asset is deeper in the money.

\begin{figure}[!ht]
\centering
\includegraphics[width=0.6\textwidth]{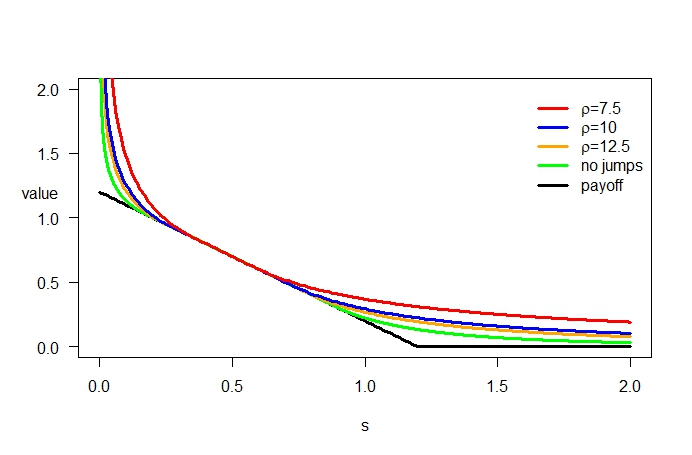}
\caption{\footnotesize{The American put option value for different intensities of exponential distribution. Parameters: $\sigma=0.2, \lambda=0.2,\ \mu=0.06$, $q=-0.01$.}}
\label{fig:ex_jumps2}
\end{figure}

\begin{figure}[!ht]
\centering
\includegraphics[width=0.6\textwidth]{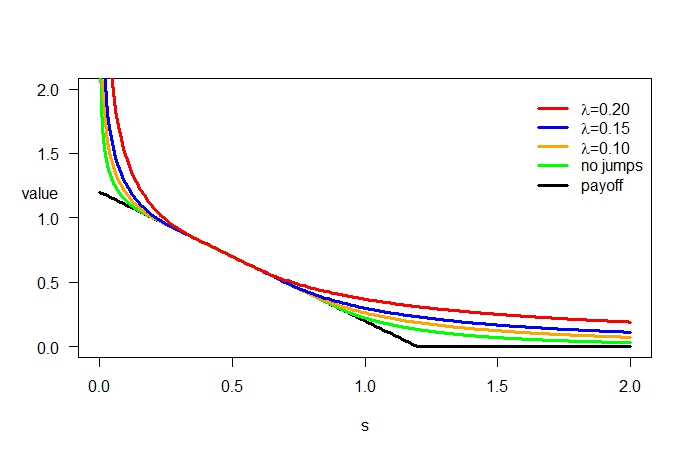}
\caption{\footnotesize{The American put option value for different intensities of exponential distribution. Parameters: $\sigma=0.2, \rho=7.5,\ \mu=0.06$, $q=-0.01$.}}
\label{fig:ex_jumps_lambda}
\end{figure}

Finally, we also compare the cases of spectrally negative and spectrally positive processes. When we set just opposite sign of the jumps comparing to spectrally negative case
and all other parameters constant, we observe larger stopping region and lower value function in case of positive jumps. This behaviour is consistent with previous analysis of jump's parameters $\rho$ and $\lambda$. When the jumps for spectrally negative process get smaller, i.e. the jump intensity $\rho$ increases, the stopping region expands. If we allow opposite sign of jump's size, we get a spectrally positive process and further dilation. This result can be observed in Figure \ref{fig:ex_jumps_positive}.

\begin{figure}[!ht]
\centering
\includegraphics[width=0.6\textwidth]{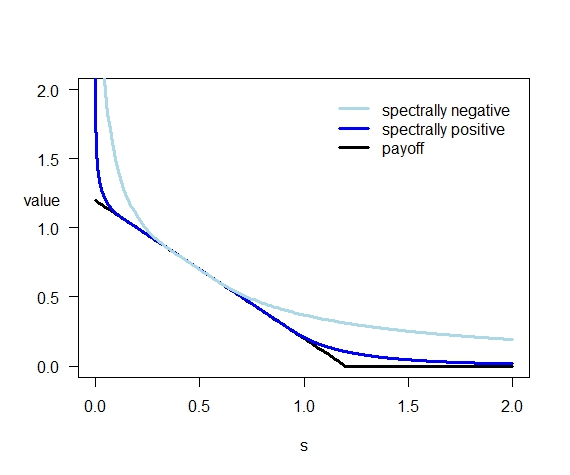}
\caption{\footnotesize{The American put option value for spectrally asymmetric processes. Parameters: $\sigma=0.2,\ \rho=7.5,\ \lambda=0.2,\ \mu=0.06$, $q=-0.01$.}}
\label{fig:ex_jumps_positive}
\end{figure}

\bigskip
\section{Conclusions and future work}\label{sec:conclusions}
In this paper we  have studied perpetual American options in a completely asymmetric L\'evy market, when  the discounting rate is negative.
We have explicitly calculated the price of the American options and identified the critical prices, hence the double continuation region, in terms of the scale function of an appropriate spectrally negative L\'evy process. We have also extended our approach to the analysis of Swing options.
We have also conducted an extensive numerical analysis for the Black-Scholes model and
the jump-diffusion model with exponential jumps. The main vehicle for deriving all results was based on
fluctuation theory of L\'evy processes and general theory of stopping problems.
This kind of study is still at its first step and several further extensions can be done.
First of all, the characterization of the price and critical prices of the perpetual American options is the starting point for the analysis of finite-maturity American options and Canadized options.
Moreover, one can think of more general settings involving Markov regime switching
to model the underlying price. These type of  results  could be achieved using the work of 
\cite{ivanovs}
where some exit identities related to the first passage times \eqref{exittimes} are given.
It would also be interesting to consider multi-dimensional exponential L\'evy process
(see 
\cite{Klimsiak}).
Finally,  the analysis  of a market in the presence of a negative discounting rate should address also other types of options, e.g. the $\pi$-options
described by 
\cite{pi}.
All of these problems are quite complex and left for future research.

\section*{Acknowledgments}
This paper was conceived while the first author was visiting the Hugo Steinhaus Center at Wroc\l aw University of Science and Technology. A special thank is due to prof. A. Weron and his group  for their kind hospitality.

\bibliographystyle{plain}

\end{document}